\documentclass[11pt]{article}
\pdfoutput=1
\usepackage{jheppub}

\usepackage{blindtext}

\usepackage{mathrsfs}
\usepackage{bm}
\usepackage{paralist}
\usepackage{url}
\usepackage[]{microtype}

\usepackage{tikz-cd}
\usepackage{tikz}
\usepackage{pict2e}
\usepackage{float}
\usepackage{stmaryrd}
\usepackage{geometry}
\usepackage{mathtools}

\usepackage{amsmath,amssymb,amsthm,amscd,graphicx}
\usepackage{psfrag}

\usepackage[utf8]{inputenc}
\usepackage{graphicx}
\usepackage{xcolor}
\usepackage{hyperref}
\hypersetup{colorlinks,allcolors=blue}

\definecolor{frangreen}{rgb}{0.040, 0.475, 0.435}

\usepackage[T1]{fontenc}

\usepackage[utf8]{inputenc}
\usepackage{graphicx}
\usepackage{amsmath,amsthm,amssymb,physics,mathtools}
\usepackage{tikz,tikz-cd, dynkin-diagrams}
\usepackage{epigraph}
\usepackage{hhline}
\usepackage{caption}
\usepackage{subcaption}
\usepackage{ytableau}
\usepackage{natbib}

\usepackage{environ}         
\usepackage{etoolbox}        
\usepackage{graphicx}        
\usepackage{nomencl}
\makenomenclature

	\newlength{\myl}
\let\origequation=\equation
\let\origendequation=\endequation

\RenewEnviron{equation}{
  \settowidth{\myl}{$\BODY$}                       
  \origequation
  \ifdimcomp{\the\linewidth}{>}{\the\myl}
  {\ensuremath{\BODY}}                             
  {\resizebox{\linewidth}{!}{\ensuremath{\BODY}}}  
  \origendequation
}

\newtheorem{teo}{Theorem}[section]
\newtheorem{prop}[teo]{Proposition}
\newtheorem{lemma}[teo]{Lemma}
\newtheorem{corol}[teo]{Corollary}

\newtheorem{remark}{Remark}[section]
\theoremstyle{definition}
\newtheorem{defin}[teo]{Definition}

\newcommand{\bea}{\begin{eqnarray}}
\newcommand{\eea}{\end{eqnarray}}

\def\v{\v}

\def\12{\frac{1}{2}}
\newcommand{\be}{\begin{equation}}
\newcommand{\ee}{\end{equation}}
\newcommand{\ba}{\begin{aligned}}
	\newcommand{\ben}{\begin{eqnarray}\displaystyle}
	\newcommand{\een}{\end{eqnarray}}

\makeatletter
\gdef\@fpheader{}
\makeatother
	
\begin{document}

		\title{On the convergence of Nekrasov functions}

		\author{Paolo Arnaudo, Giulio Bonelli, Alessandro Tanzini}
		\affiliation{International School of Advanced Studies (SISSA), via Bonomea 265, 34136 Trieste, Italy}
		\affiliation{INFN, Sezione di Trieste, Trieste, Italy}
		\affiliation{Institute for Geometry and Physics, IGAP, via Beirut 2, 34136 Trieste, Italy}
    \emailAdd{parnaudo@sissa.it}
    \emailAdd{bonelli@sissa.it}
    \emailAdd{tanzini@sissa.it}
		
		\date{\today}

\abstract{
		In this note we present some results on the convergence of Nekrasov 
partition functions as power series in the instanton counting parameter.
We focus on $U(N)$ ${\mathcal N}=2$ gauge theories in four dimensions with matter in the adjoint and in the fundamental representations of the gauge group respectively and find rigorous lower bounds for the convergence radius in the two cases:
if the theory is {\it conformal}, then the series has at least a {\it finite} radius of convergence, while if it is
{\it asymptotically free} it has {\it infinite} radius of convergence. Via AGT correspondence, this implies that the related irregular conformal blocks of $W_N$ algebrae admit 
a power expansion in the modulus converging in the whole plane. By specifying to the $SU(2)$ case, we apply our 
results to analyse the convergence properties of the corresponding Painlev\'e $\tau$-functions.
}
	\maketitle
	
\section{Introduction}

The study of the analytic structure in the coupling constant of partition functions and correlation functions in quantum field theory (QFT) is a theme of paramount importance in this subject since its birth
\cite{Dyson,tHooft:1977xjm}.
Having control on the exact general coefficient of these complete expansions is very rare, but nonetheless crucial to study in depth quantum field theories beyond their incomplete perturbative definition. To this end the control on non-perturbative effects is very important.

The aim of this paper is to 
focus on the convergence properties of instanton effects in
a case in which these can be calculated explicitly.
This happens for BPS partition functions of supersymmetric enough QFTs so that, 
due to no renormalisation theorems, higher loop corrections vanish because of boson-fermion cancellation. This happens for ${\mathcal N}=2$ gauge theories in $D=4$ where
 the coefficients of the instanton expansion are exactly known, so that one can precisely estimate the radius of convergence of the series.

The resulting instanton series, known as Nekrasov functions
\cite{Nekrasov:2002qd}, have many applications in mathematical physics problems, ranging from
quantization of integrable systems
\cite{Nekrasov:2009rc}, relation to conformal blocks of Virasoro algebra
\cite{Alday:2009aq}, isomonodromic deformation theory 
\cite{Gamayun:2012ma,Bonelli:2016qwg},
non perturbative approaches to the quantisation of integrable systems
\cite{Grassi:2019coc,Fioravanti:2019vxi}
and Heun functions \cite{Bonelli:2022ten,Lisovyy:2022flm}. 
Also from these perspectives, a rigorous analysis of the convergence properties of the Nekrasov functions is important.  

In this paper we prove some theorems on the convergence of Nekrasov 
functions as power series in the complexified gauge coupling
$q=e^{2\pi i\tau}$, with $\tau=\frac{\theta}{2\pi}+i\frac{4\pi}{g^2}$, as $Z_{\rm inst}=\sum_{k\geq 0}q^k z_k$.
The general coefficient $z_k$ of the series for the theory with $U(N)$ gauge group is 
the equivariant volume of the moduli space $\mathcal{M}_{k,N}$ of
k-instantons in $U(N)$ gauge theory. This can be computed via equivariant localisation formulae as a sum over the fixed points of the algebraic torus action $\left(\mathbb{C}^*\right)^{N+2}$ on $\mathcal{M}_{k,N}$. From the gauge theory viewpoint, the associated equivariant weights $a_i, \, i=1,\ldots, N$ and 
$\epsilon_1,\epsilon_2$ are
 the vevs of the Higgs field and the parameters of the so-called $\Omega$-background respectively.
 The fixed points are classified by {\it coloured partitions} of $k$, described by a collection of  $N$ Young diagrams
with total number of boxes $k$. The equivariant parameters of line bundles over $\mathcal{M}_{k,N}$
describe the masses of the matter content of the theory. The coefficients $z_k$ turns out to be 
rational functions of the equivariant parameters. 
Therefore, to analyse the convergence properties of the partition functions, one needs to 
estimate uniformly in the equivariant parameters the behaviour at large $k$ of these intricate rational functions.
We will adapt to the case at hand some known combinatorial results to prove
some -- physically meaningful -- estimates for the convergence radius of these partition functions.
In short, we will prove that
\begin{itemize}
    \item if the theory is {\it asymptotically free}, then the multiinstanton series has {\it infinite} radius of convergence;
    \item if the theory is {\it conformal}, then the multiinstanton series has at least a {\it finite} radius of convergence.
\end{itemize}

More precisely, we can establish the above results in a number of cases and under some genericity assumptions on the background parameters.
We analyse the cases of the gauge theory where an ${\mathcal N}=2$ $U(N)$ vector multiplet 
is coupled to a massive hypermultiplet in the adjoint representation (that is ${\mathcal N}=2^*$ theory) and the case in which it is coupled to $N_f\leq 2N$ hypermultiplets in the fundamental.
These two cases are respectively split between sections 2 and 3 of the paper.

Natural genericity assumptions on the Higgs field vacuum expectation values -- which are taken to be immeasurable with respect to the $\Omega$-background parameters -- 
are imposed in order to avoid potential poles in the rational functions coefficients $z_k$, while mass parameters are not constrained.
In section 2 we analyse ${\mathcal N}=2^*$ theory by assuming generic  
$\Omega$-background parameters, while
in section 3 the study of the gauge theory with matter in the fundamental representation of the gauge group is performed in the self-dual $\Omega$-background 
$\epsilon_1+\epsilon_2=0$.

In section 2 we prove Theorem \ref{teoN=2} stating that the instanton partition function of the $\mathcal{N}=2$ $U(N)$ gauge theory with an adjoint multiplet with mass $m$
and ${\rm Arg}(\epsilon_1/\epsilon_2)\neq0$
is a power series converging absolutely within the disk
\begin{equation}\label{th2}
|q|<\left(1+\frac{|m|}{D(\vec{a},\epsilon_1,\epsilon_2)}\right)^{-2(N-1)},
\end{equation}
where 
$$
D(\vec{a},\epsilon_1,\epsilon_2)=\min_{1\le i\ne j\le N}\{\min_{p\in
{\mathbb Z}\epsilon_1+{\mathbb Z}\epsilon_2}\{|a_i-a_j-p|\}\}.
$$

In section 3 we consider the instanton partition function of the $\mathcal{N}=2$ $U(N)$ gauge theory with $N_f\leq2N$ multiplets in the fundamental representation and self-dual $\Omega$-background, with $\epsilon_1=-\epsilon_2=\epsilon$. For $N_f=2N$ we prove that,
as a power series in $q$, it
converges absolutely within 
a disk whose radius depends on the values of the Higgs 
vacuum expectation values $a_i=\epsilon\alpha_i$
as
\begin{equation}
|q|<\prod_{\{(i,j)\in\{1,\dots,N\}^2\ |\ i\ne j\}}\left[\frac{2^4}{\min\{1,|\alpha_i-\alpha_j|\}}\left( 1+\frac{|\alpha_i-\alpha_j|}{C_{ij}(\vec{\alpha})}\right)\right]^{-1}\, ,
\end{equation}
where $$C_{ij}(\vec{\alpha})=\min_{n\in\mathbb{Z}}|\alpha_i-\alpha_j-n|.$$
This is the main content of Theorem \ref{teoremafondamentale}.
By standard holomorphic decoupling, one finds that asymptotically free cases $N_f<2N$ have instead infinite radius of convergence as power series in the corresponding renormalisation group invariant scale.

In section 4 we apply our results to analyse the convergence properties of some Painlev\'e $\tau$-functions.
The Kiev formula conjectured in \cite{Gamayun:2013auu}
states that Painlev\'e $\tau$-functions can be expressed as discrete Fourier transforms of 
suitable full $SU(2)$ Nekrasov partition functions. Generically, 
$$\tau_P(q;a,s)\propto\sum_{n\in{\mathbb Z}}s^n q^{(a+n)^2}Z_{{\rm 1loop}}(a+n)Z_{{\rm inst}}(q,a+n)$$
where $Z_{{\rm 1loop}}(a)$ is the perturbative 1-loop contribution while $a$ and $s$ parametrize the initial conditions of the Painlev\'e flow.
The study of the convergence of the Fourier series is done applying the above theorems and shows 
that the absolute convergence bounds on $Z_{{\rm inst}}(q,a)$ extend to $\tau_P(q;a,s)$. 
Our result generalises the one for PIII$_3$ equation obtained in \cite{Its:2014lga}.

By AGT correspondence \cite{Alday:2009aq}, the instanton partition function is identified with the Virasoro,
and more in general $W$-algebras \cite{Wyllard:2009hg}, 
conformal blocks\footnote{
The Liouville central charge is $c=1+6Q^2$, $Q=b+1/b$, 
$b=\sqrt{\frac{\epsilon_1}{\epsilon_2}}$. Liouville momenta are parametrised as 
$\alpha=Q/2+ia$, where $a\in{\mathbb R}$ is a  Cartan or a mass parameter.}. 
In particular, the ${\mathcal N}=2^*$
theory corresponds to one point conformal blocks on the torus, while the ${\mathcal N}=2$ with $N_f=2N$ fundamentals corresponds to four point 
conformal blocks on the Riemann sphere. The cases $N_f<2N$
instead involve correlators with $W_N$ irregular states
\cite{Gaiotto:2009ma,Bonelli:2011aa,Gaiotto:2012sf,Kanno:2013vi}. 
In this context our results provide a lower bound for the convergence radius of the regular conformal blocks of $W$-algebrae and establish that  irregular conformal blocks expansion has
an infinite radius of convergence.
We remark that the latter result is in line with the expectations coming from special choices of external momenta for which the correlator reduces to known functions \cite{Lisovyy:2018mnj}. For the regular conformal blocks analogous reasoning and modular invariance of the correlators leads to the expectation that the radius of convergence is actually one. Our estimate from {\it brute force} direct inspection of the combinatorial formulae for the coefficients is therefore not optimal and could be hopefully improved using other arguments related to S-duality properties of the corresponding supersymmetric gauge theories.
Let us remark that some results for the particular case of Virasoro algebra 
and non generic background parameters have been recently derived by using a probabilistic approach\footnote{See also 
\cite{REMY}
at 
\url{http://www.math.columbia.edu/~remy/files/Modular_Equation.pdf}} 
in \cite{Colin,Ghosal:2020}.
Previous studies on the convergence radius of the instanton series for pure $\mathcal{N}=2$  Super Yang-Mills  with $SU(2)$ gauge group appeared in \cite{Its:2014lga} which discussed the $\epsilon_1+\epsilon_2=0$ case for the four dimensional gauge theory, and in \cite{Bershtein:2016aef,Felder:2017rgg} where the case of five dimensional gauge theory on $\mathbb{R}^4\times S^1$ was addressed for non-generic values of $\Omega$-background parameters.

\vspace{.3cm}

There are several interesting questions to be further investigated. 

Notice that the bound we found after the proof of theorem \ref{teoN=2}
regularly extends to the excluded ray ${\rm Arg}(\epsilon_1/\epsilon_2)=0$.
It is therefore conceivable that it could be proved, with other techniques,
for any non zero value of the $\Omega$-background parameters.

For technical reasons the case of fundamental matter was analysed for $\epsilon_1+\epsilon_2=0$ 
$\Omega$-background, but we believe that our results can be extended also to generic values of the $\epsilon$s. For example, one should be able to extend formula
\eqref{polynomialidentity} in order to study the case $\epsilon_1/\epsilon_2\in {\mathbb Q_{<0}}$.

It would be also very interesting to 
be able to extend the study of Nekrasov function combinatorics in the Nekrasov-Shatashvili limit $\epsilon_1=0$ \cite{Nekrasov:2009rc},
corresponding to the classical limit
of conformal blocks.

Regarding the above issues, complementing our analysis with blow-up equations for the Nekrasov partition
functions \cite{Gottsche:2010ig}  could improve our results. 

The results we obtained provide an explanation of the unreasonable effectiveness of instanton counting also in strongly coupled phases, such as the $\mathcal{N}=1$ confining vacua \cite{Fucito:2005wc}. One would be then tempted to apply to Argyres-Douglas superconformal points \cite{Argyres:1995jj}, with the caveat that these are reached via a double scaling limit where the parameters appearing in the coefficients of the instanton series, in particular the vevs of the Higgs field, are redefined. We did not address this issue in this paper.

One obvious extension of our analysis is to linear and circular quiver gauge theories in general $\Omega$-background, which, on the two-dimensional CFT counterpart, correspond to conformal blocks with several insertions on the sphere and on the torus, respectively. 

It would also be interesting to extend the approach and results of this paper 
to the corresponding five dimensional gauge theories on a circle.

As mentioned above, it would be interesting to complement our analysis 
basing on electromagnetic duality and the induced modular properties of the BPS partition functions, from which one expects
that the radius of convergence for the conformal
cases is $|q|<1$ for generic values of the mass parameters. 
For an exhaustive analysis of their analytic properties, one could exploit the relation of these partition functions 
with the solutions of isomonodromic deformation problems (a.k.a. differential equations of Painlev\'e type) and their singularity theory.
Let us remark that Fredholm  determinant representations of the related $\tau$-functions have been derived \cite{Gavrylenko:2016zlf} and their analysis could be used to provide a proof of the desired convergence properties.

\vskip 1cm
{\bf Acknowledgments}:
We would like to thank 
M. Bershtein, H. Desiraju, G. Felder, A. Grassi, O. Lisovyy, G. Remy and A. Shchechkin
for useful questions, discussions and clarifications. 
This research is partially supported by the INFN Research Projects GAST and ST\&FI, PRIN "Geometria delle varietà algebriche", PRIN "Non-perturbative Aspects Of Gauge Theories And Strings", PRIN "String Theory as a bridge between Gauge Theories and Quantum Gravity" and INdAM.

\section{Convergence of $U(N)$ Instanton Partition Function with adjoint matter}

We begin our analysis with the study of the convergence properties of the instanton partition function of $\mathcal{N}=2^*$ $U(N)$ gauge theory 
\begin{equation}\label{ciao}
\begin{aligned}
&Z_{\mathrm{inst}}^{\mathcal{N}=2^*, U(N)}=
\sum_{k\ge 0}q^k\sum_{|\vec{Y}|=k}\prod_{i=1}^N\prod_{s\in Y_i}\left(1-\frac{m}{-\epsilon_1L_{Y_i}(s)+\epsilon_2(A_{Y_i}(s)+1)}\right)\left(1-\frac{m}{\epsilon_1(L_{Y_i}(t)+1)-\epsilon_2A_{Y_i}(s)}\right)\\&\prod_{1\le i\ne j\le N}\prod_{s\in Y_i}\left(1-\frac{m}{a_i-a_j-\epsilon_1L_{Y_j}(s)+\epsilon_2(A_{Y_i}(s)+1)}\right)
\prod_{t\in Y_j}\left(1-\frac{m}{-a_j+a_i+\epsilon_1(L_{Y_i}(t)+1)-\epsilon_2A_{Y_j}(t)}\right).
\end{aligned}
\end{equation}
We refer to Appendix \ref{appendixA} for the notations used.
In the products above we collected first the pairs with $i=j$ (in what follows we will call these contributions \emph{diagonal}), and then the pairs $(i,j)$ with $i\ne j$ (in what follows we will call these contributions \emph{nondiagonal}).
From a direct inspection of \eqref{ciao}, one can see that the coefficients of the series are well defined under the assumptions
\begin{equation}\label{conditions}
\mathrm{Arg}\left(\frac{\epsilon_2}{\epsilon_1}\right)\ne 0 \quad {\rm and} \quad\pm(a_i-a_j)\notin\Lambda(\epsilon_1,\epsilon_2)\quad
\forall 1\le i<j\le N\, ,
\end{equation} where $\Lambda(\epsilon_1,\epsilon_2)$ is the 2-dimensional lattice
\begin{equation}
\Lambda(\epsilon_1,\epsilon_2)=\{z\in\mathbb{C}\ | z\in\epsilon_1\mathbb{Z}+\epsilon_2\mathbb{Z}\},
\end{equation}
which we will use in the proof of the 
\begin{teo}\label{teoN=2}
The instanton partition function of the $\mathcal{N}=2^*$ $U(N)$ gauge theory, as a power series in the complex parameter $q$, is absolutely convergent at least for
\begin{equation}
|q|<\left(1+\frac{|m|}{D(\vec{a},\epsilon_1,\epsilon_2)}\right)^{-2(N-1)},
\end{equation}
where $m$ is the mass of the adjoint multiplet, and 
\begin{equation}
D(\vec{a},\epsilon_1,\epsilon_2)=\min_{1\le i\ne j\le N}\{\min_{p\in\Lambda(\epsilon_1,\epsilon_2)}\{|a_i-a_j-p|\}\}.
\end{equation}
\end{teo}

From this result, two corollaries can be proved. The first comes from the fact that the $\mathcal{N}=2^*$ instanton partition function reduces to the $\mathcal{N}=2$ SYM instanton partition function in the double scaling limit $q\to 0$ and $m\to\infty$ with $\Lambda=qm^{2N}$ kept finite.
\begin{corol}\label{corolN=2}
The instanton partition function of the $U(N)$ pure gauge theory, as a power series in the complex parameter $\Lambda$, is convergent over the whole complex plane.
\end{corol}
The second corollary comes from the fact that if the mass of the adjoint multiplet goes to zero, $m\to 0$, the $\mathcal{N}=2^*$ instanton partition function reduces to the $\mathcal{N}=4$ instanton partition function.
\begin{corol}\label{corolN=4}
The instanton partition function of the $\mathcal{N}=4$ $U(N)$ gauge theory, as a power series in the complex parameter $q$, is convergent in the region $|q|< 1$.
\end{corol}

\begin{remark}
By using
known analytic properties of the partition function \eqref{ciao}, 
one can lift \eqref{conditions} to milder conditions for the values of the $a$-parameters. 
Indeed, the second condition, which 
we imposed to a priori get rid of the possible poles in the non-diagonal part, can be 
reduced to the set of actual poles as classified in 
\cite{Zamolodchikov:1984eqp, Poghossian:2017atl,
Sysoeva:2022syp}.
\end{remark}

\begin{remark}
The content of Corollary \ref{corolN=2} is an higher rank generalisation of an observation about the $SU(2)$ SYM ${\mathcal N=2}$ instanton partition function given in 
\cite{Its:2014lga}.
\end{remark}

\begin{remark}
Corollary \ref{corolN=4} is trivial. Indeed, it is well known that the ${\mathcal N=4}$ partition function is equal to $\phi(q)^{-N}$, $\phi(q)$
being the Euler function. 
\end{remark}

\subsection{Proof of Theorem \ref{teoN=2}}

It will be useful to divide the assumption
\begin{equation}
\mathrm{Arg}\left(\frac{\epsilon_2}{\epsilon_1}\right)\ne 0
\end{equation}
in the two  subcases:
\begin{enumerate}
\item $\mathrm{Im}\bigl(\frac{\epsilon_2}{\epsilon_1}\bigr)\ne 0;$
\item $\mathrm{Re}\bigl(\frac{\epsilon_2}{\epsilon_1}\bigr)<0.$
\end{enumerate}

\subsubsection{First Subcase}

Suppose $$\mathrm{Im}\left(\frac{\epsilon_2}{\epsilon_1}\right)\ne 0.$$
Let $\delta>0$ be a real number such that $$\min\biggl\{\bigg|\mathrm{Im}\biggl(\frac{\epsilon_2}{\epsilon_1}\biggr)\bigg|,\bigg|\mathrm{Im}\biggl(\frac{\epsilon_1}{\epsilon_2}\biggr)\bigg|\biggr\}>\delta.$$
Notice that  $\delta\le 1$. 

We first analyze the products over the boxes of one of the diagrams, say $Y_1$, whose contributions come from the diagonal factors, namely we look for a bound on 
\begin{equation}
\bigg|\prod_{s\in Y_1}\left(1-\frac{m}{-\epsilon_1L_{Y_1}(s)+\epsilon_2(A_{Y_1}(s)+1)}\right)\left(1-\frac{m}{\epsilon_1(L_{Y_1}(t)+1)-\epsilon_2A_{Y_1}(s)}\right)\bigg|.
\end{equation}
An analogous reasoning will also hold for the diagonal contributions of the other diagrams $Y_2,\dots,Y_N$. 

We begin by estimating the denominators in the previous product. Let us fix a box $s\in Y_1$, and let us consider the term
\begin{equation}
\frac{1}{|-\epsilon_1L_{Y_1}(s)+\epsilon_2(A_{Y_1}(s)+1)|}.
\end{equation}
By recalling the definition of {\it hook length}, 
$h_{Y_1}(s)=L_{Y_1}(s)+A_{Y_1}(s)+1$, we can without loss of generality suppose $$A_{Y_1}(s)\ge\frac{h_{Y_1}(s)-1}{2}.$$ Then, if we collect a factor of $\epsilon_1$, we have 
\begin{equation}
\begin{aligned}
|-\epsilon_1L_{Y_1}(s)+\epsilon_2(A_{Y_1}(s)+1)|&=|\epsilon_1|\cdot|L_{Y_1}(s)-\frac{\epsilon_2}{\epsilon_1}(A_{Y_1}(s)+1)|\ge |\epsilon_1|\cdot\bigg|\mathrm{Im}\biggl(\frac{\epsilon_2}{\epsilon_1}\biggr)\bigg|\cdot(A_{Y_1}(s)+1)\\
&\ge |\epsilon_1|\cdot\delta\cdot\frac{h_{Y_1}(s)+1}{2}\ge |\epsilon_1|\cdot\delta\cdot\frac{h_{Y_1}(s)}{4}.
\end{aligned}
\end{equation}
Analogously, for the term
\begin{equation}
\frac{1}{|\epsilon_1(L_{Y_1}(s)+1)-\epsilon_2A_{Y_1}(s)|},
\end{equation}
we have 
\begin{equation}
\begin{aligned}
|\epsilon_1(L_{Y_1}(s)+1)-\epsilon_2A_{Y_1}(s)|&=|\epsilon_1|\cdot|L_{Y_1}(s)+1-\frac{\epsilon_2}{\epsilon_1}A_{Y_1}(s)|\ge |\epsilon_1|\cdot\bigg|\mathrm{Im}\biggl(\frac{\epsilon_2}{\epsilon_1}\biggr)\bigg|\cdot A_{Y_1}(s)\\
&\ge |\epsilon_1|\cdot \delta\cdot\frac{h_{Y_1}(s)-1}{2}.
\end{aligned}
\end{equation}
Notice that, if $h_{Y_1}(s)=1$, then both $L_{Y_1}(s)=A_{Y_1}(s)=0$, and the previous term is simply $|\epsilon_1|=|\epsilon_1|h_{Y_1}(s)\ge |\epsilon_1|\cdot\delta\cdot\frac{h_{Y_1}(s)}{4}$, and, if $h_{Y_1}(s)\ge 2$, then $(h_{Y_1}(s)-1)/2\ge h_{Y_1}(s)/4$. Therefore, also this term is always bounded by $|\epsilon_1|\cdot \delta\cdot\frac{h_{Y_1}(s)}{4}$.

If we instead considered a box $s$ for which $$A_{Y_1}(s)<\frac{h_{Y_1}(s)-1}{2},$$ we would have $$L_{Y_1}(s)\ge\frac{h_{Y_1}(s)-1}{2}.$$ In this case, we collect factors of $\epsilon_2$ from both terms to obtain
\begin{equation}
\begin{aligned}
|-\epsilon_1L_{Y_1}(s)+\epsilon_2(A_{Y_1}(s)+1)|&\ge |\epsilon_2|\cdot\delta\cdot\frac{h_{Y_1}(s)}{4},\\
|\epsilon_1(L_{Y_1}(s)+1)-\epsilon_2A_{Y_1}(s)|&\ge |\epsilon_2|\cdot\delta\cdot\frac{h_{Y_1}(s)}{4}.
\end{aligned}
\end{equation}
Now, fix 
\begin{equation}
|\epsilon|=\min\{|\epsilon_1|,|\epsilon_2|\}.
\end{equation}
Then,
\begin{equation}\label{diagonalperentrambi}
\begin{aligned}
\frac{1}{|-\epsilon_1L_{Y_1}(s)+\epsilon_2(A_{Y_1}(s)+1)|}&\le\frac{4}{|\epsilon|\cdot\delta\cdot h_{Y_1}(s)},\\
\frac{1}{|\epsilon_1(L_{Y_1}(s)+1)-\epsilon_2A_{Y_1}(s)|}&\le \frac{4}{|\epsilon|\cdot\delta\cdot h_{Y_1}(s)}.
\end{aligned}
\end{equation}
Therefore, we have that
\begin{equation}\label{diagonal_bound}
\begin{aligned}
&\bigg|\prod_{s\in Y_1}\biggl(1-\frac{m}{-\epsilon_1 L_{Y_1}(s)+\epsilon_2(A_{Y_1}(s)+1)}\biggr)\biggl(1-\frac{m}{\epsilon_1(L_{Y_1}(t)+1)-\epsilon_2 A_{Y_1}(t)}\biggr)\bigg|=\\
&\bigg|\prod_{s\in Y_1}\biggl(1+\frac{m^2-m(\epsilon_1+\epsilon_2)}{[-\epsilon_1 L_{Y_1}(s)+\epsilon_2(A_{Y_1}(s)+1)][\epsilon_1(L_{Y_1}(t)+1)-\epsilon_2 A_{Y_1}(t)]}\biggr)\bigg|\le\\
&\prod_{s\in Y_1}\biggl(1+\frac{16|m^2-m(\epsilon_1+\epsilon_2)|}{\delta^2|\epsilon|^2h_{Y_1}(s)^2}\biggr)=\prod_{s\in Y_1}\left(1+\frac{\frac{16|m^2-m(\epsilon_1+\epsilon_2)|}{\delta^2|\epsilon|^2}}{h_{Y_1}(s)^2}\right).
\end{aligned}
\end{equation}

We now consider the remaining terms, which come from the nondiagonal contributions. We analyze the products over the boxes of $Y_1$, coming from the pairs $(1,2)$ and $(2,1)$. The products over the boxes of the diagram $Y_2$ in the same pairs will be analogous, and the same will hold for any other couple of pairs $(i,j),(j,i)$.

The terms we consider are then
\begin{equation}
\bigg|\prod_{s\in Y_1}\biggl(1-\frac{m}{a_1-a_2-\epsilon_1 L_{Y_2}(s)+\epsilon_2(A_{Y_1}(s)+1)}\biggr)\biggl(1-\frac{m}{-a_1+a_2+\epsilon_1(L_{Y_2}(s)+1)-\epsilon_2 A_{Y_1}(s)}\biggr)\bigg|.
\end{equation}
With our assumptions on the vev parameters $a_i$ we have that the denominators are never zero. Let us define
\begin{equation}
D_{ij}(\vec{a},\epsilon_1,\epsilon_2)=\min_{p\in\Lambda(\epsilon_1,\epsilon_2)}\{|a_i-a_j-p|\}.
\end{equation}
We have that
\begin{equation}\label{nondiagonal_bound}
\begin{aligned}
&\bigg|\prod_{s\in Y_1}\biggl(1-\frac{m}{a_1-a_2-\epsilon_1 L_{Y_2}(s)+\epsilon_2(A_{Y_1}(s)+1)}\biggr)\biggl(1-\frac{m}{-a_1+a_2+\epsilon_1(L_{Y_2}(s)+1)-\epsilon_2 A_{Y_1}(s)}\biggr)\bigg|\le\\
&\prod_{s\in Y_1}\biggl(1+\frac{|m|}{|a_1-a_2-\epsilon_1 L_{Y_2}(s)+\epsilon_2(A_{Y_1}(s)+1)|}\biggr)\biggl(1+\frac{|m|}{|-a_1+a_2+\epsilon_1(L_{Y_2}(s)+1)-\epsilon_2 A_{Y_1}(s)|}\biggr)\le\\
&\left(1+\frac{|m|}{D_{12}(\vec{a},\epsilon_1,\epsilon_2)}\right)^{2|Y_1|}.
\end{aligned}
\end{equation}

Putting the bounds \eqref{diagonal_bound} and \eqref{nondiagonal_bound} together, we can conclude
\begin{equation}
\begin{aligned}
&\bigg|\sum_{k\ge 0}q^k\sum_{|\vec{Y}|=k}\prod_{i=1}^N\prod_{s\in Y_i}\left(1-\frac{m}{-\epsilon_1L_{Y_i}(s)+\epsilon_2(A_{Y_i}(s)+1)}\right)\left(1-\frac{m}{\epsilon_1(L_{Y_i}(t)+1)-\epsilon_2A_{Y_i}(s)}\right)\times\\
&\ \  \prod_{1\le i\ne j\le N}\prod_{s\in Y_i}\left(1-\frac{m}{a_i-a_j-\epsilon_1L_{Y_j}(s)+\epsilon_2(A_{Y_i}(s)+1)}\right)\prod_{t\in Y_j}\left(1-\frac{m}{-a_j+a_i+\epsilon_1(L_{Y_i}(t)+1)-\epsilon_2A_{Y_j}(t)}\right)\bigg|\le\\
&\le\sum_{k\ge 0}|q|^k\sum_{|\vec{Y}|=k}\prod_{i=1}^N\prod_{s\in Y_i}\left(1+\frac{\frac{16|m^2-m(\epsilon_1+\epsilon_2)|}{\delta^2|\epsilon|^2}}{h_{Y_i}(s)^2}\right)\prod_{i=1}^N\prod_{j\ne i}\left(1+\frac{|m|}{D_{ij}(\vec{a},\epsilon_1,\epsilon_2)}\right)^{2|Y_i|}.
\end{aligned}
\end{equation}
Now, let us define
\begin{equation}
D(\vec{a},\epsilon_1,\epsilon_2)=\min_{1\le i\ne j\le N}\{D_{ij}(\vec{a},\epsilon_1,\epsilon_2)\}.
\end{equation}
The following result (which is Theorem 1.2 in \cite{nekrasovokounkov}) will be useful:
\begin{prop}
For any complex number $z$ the following holds:
\begin{equation}\label{teo1.2}
\sum_{Y\in\mathbb{Y}}x^{|Y|}\prod_{s\in Y}\biggl(1-\frac{z}{(h_Y(s))^2}\biggr)=\prod_{j= 1}^{\infty}(1-x^j)^{z-1}=\phi(x)^{z-1},
\end{equation}
where $\phi$ is the Euler function.
\end{prop}
We remind that $\phi(x)$ is convergent for $|x|<1$.
Then,
\begin{equation}
\begin{aligned}
&\sum_{k\ge 0}|q|^k\sum_{|\vec{Y}|=k}\prod_{i=1}^N\prod_{s\in Y_i}\left(1+\frac{\frac{16|m^2-m(\epsilon_1+\epsilon_2)|}{\delta^2|\epsilon|^2}}{h_{Y_i}(s)^2}\right)\prod_{i=1}^N\prod_{j\ne i}\left(1+\frac{|m|}{D_{ij}(\vec{a},\epsilon_1,\epsilon_2)}\right)^{2|Y_i|}\le\\
&\sum_{k\ge 0}\left[|q|\left(1+\frac{|m|}{D(\vec{a},\epsilon_1,\epsilon_2)}\right)^{2(N-1)}\right]^k\sum_{|\vec{Y}|=k}\prod_{i=1}^N\prod_{s\in Y_i}\left(1+\frac{\frac{16|m^2-m(\epsilon_1+\epsilon_2)|}{\delta^2|\epsilon|^2}}{h_{Y_i}(s)^2}\right)=\\
&\sum_{Y_1,\dots,Y_N\in\mathbb{Y}}\left[|q|\left(1+\frac{|m|}{D(\vec{a},\epsilon_1,\epsilon_2)}\right)^{2(N-1)}\right]^{\sum_{i=1}^N|Y_i|}\prod_{i=1}^N\prod_{s\in Y_i}\left(1+\frac{\frac{16|m^2-m(\epsilon_1+\epsilon_2)|}{\delta^2|\epsilon|^2}}{h_{Y_i}(s)^2}\right)=\\
&\left\{\sum_{Y\in\mathbb{Y}}\left[|q|\left(1+\frac{|m|}{D(\vec{a},\epsilon_1,\epsilon_2)}\right)^{2(N-1)}\right]^{|Y|}\prod_{s\in Y}\left(1+\frac{\frac{16|m^2-m(\epsilon_1+\epsilon_2)|}{\delta^2|\epsilon|^2}}{h_{Y}(s)^2}\right)\right\}^N=\\
&\phi\left(|q|\left(1+\frac{|m|}{D(\vec{a},\epsilon_1,\epsilon_2)}\right)^{2(N-1)}\right)^{N\left(-\frac{16|m^2-m(\epsilon_1+\epsilon_2)|}{\delta^2|\epsilon|^2}-1\right)},
\end{aligned}
\end{equation}
where in the last line we used \eqref{teo1.2} with $x=|q|\left(1+\frac{|m|}{D(\vec{a},\epsilon_1,\epsilon_2)}\right)^{2(N-1)}$ and $z=-\frac{16|m^2-m(\epsilon_1+\epsilon_2)|}{\delta^2|\epsilon|^2}$.

Hence, we can conclude that the instanton partition function is convergent in the region defined by
\begin{equation}
|q|< \left(1+\frac{|m|}{D(\vec{a},\epsilon_1,\epsilon_2)}\right)^{-2(N-1)}.
\end{equation}

\subsubsection{Second Subcase}

Suppose otherwise that $$\mathrm{Re}\biggl(\frac{\epsilon_2}{\epsilon_1}\biggr)<0,$$ and let $$\beta=-\mathrm{Re}\biggl(\frac{\epsilon_2}{\epsilon_1}\biggr)>0.$$
Again, we start by analyzing the products over the boxes of one of the diagrams, say $Y_1$, coming from the diagonal contributions, that is, we look for a bound on 
\begin{equation}
\bigg|\prod_{s\in Y_1}\left(1-\frac{m}{-\epsilon_1L_{Y_1}(s)+\epsilon_2(A_{Y_1}(s)+1)}\right)\left(1-\frac{m}{\epsilon_1(L_{Y_1}(t)+1)-\epsilon_2A_{Y_1}(s)}\right)\bigg|.
\end{equation}
We begin by estimating the denominators. For every box $s\in Y_1$, we have 
\begin{equation}
\begin{aligned}
&|-\epsilon_1L_{Y_1}(s)+\epsilon_2(A_{Y_1}(s)+1)|=|\epsilon_1|\cdot \bigg|L_{Y_1}(s)-\frac{\epsilon_2}{\epsilon_1}(A_{Y_1}(s)+1)\bigg|\ge\\
&|\epsilon_1|\cdot\bigg|\mathrm{Re}\left[L_{Y_1}(s)-\frac{\epsilon_2}{\epsilon_1}(A_{Y_1}(s)+1)\right]\bigg|=\\
&|\epsilon_1|\cdot\bigg|L_{Y_1}(s)-(A_{Y_1}(s)+1)\mathrm{Re}\left(\frac{\epsilon_2}{\epsilon_1}\right)\bigg|=\\
&|\epsilon_1|\cdot\bigg|L_{Y_1}(s)+\beta(A_{Y_1}(s)+1)\bigg|.
\end{aligned}
\end{equation}
Fix $\gamma=\min\{\beta,1\}$. Then,
\begin{equation}\label{fattoreuno}
\frac{1}{|-\epsilon_1L_{Y_1}(s)+\epsilon_2(A_{Y_1}(s)+1)|}\le\frac{1}{|\epsilon_1|}\frac{1}{\gamma h_{Y_1}(s)}.
\end{equation}
Analogously, for the other term in the product, we have 
\begin{equation}\label{fattoredue}
\frac{1}{|\epsilon_1(L_{Y_1}(s)+1)-\epsilon_2A_{Y_1}(s)|}\le\frac{1}{|\epsilon_1|}\frac{1}{\gamma h_{Y_1}(s)}.
\end{equation}
Then,
\begin{equation}
\begin{aligned}
&\bigg|\prod_{s\in Y_1}\left(1-\frac{m}{-\epsilon_1L_{Y_1}(s)+\epsilon_2(A_{Y_1}(s)+1)}\right)\left(1-\frac{m}{\epsilon_1(L_{Y_1}(t)+1)-\epsilon_2A_{Y_1}(s)}\right)\bigg|=\\
&\bigg|\prod_{s\in Y_1}\biggl(1+\frac{m^2-m(\epsilon_1+\epsilon_2)}{[-\epsilon_1 L_{Y_1}(s)+\epsilon_2(A_{Y_1}(s)+1)][\epsilon_1(L_{Y_1}(t)+1)-\epsilon_2 A_{Y_1}(t)]}\biggr)\bigg|\le\\
&\prod_{s\in Y_1}\left(1+\frac{|m^2-m(\epsilon_1+\epsilon_2)|}{|\epsilon_1|^2\gamma^2 h_{Y_1}(s)^2}\right)=\prod_{s\in Y_1}\left(1+\frac{\frac{|m^2-m(\epsilon_1+\epsilon_2)|}{|\epsilon_1|^2\gamma^2}}{h_{Y_1}(s)^2}\right).
\end{aligned}
\end{equation}
An analogous bound holds for all diagrams $Y_2,\dots,Y_N$. For the nondiagonal contributions we will use the same bound \eqref{nondiagonal_bound} of the previous subsection.
Therefore, in this case we have
\begin{equation}
\begin{aligned}
&\bigg|\sum_{k\ge 0}q^k\sum_{|\vec{Y}|=k}\prod_{i=1}^N\prod_{s\in Y_i}\left(1-\frac{m}{-\epsilon_1L_{Y_i}(s)+\epsilon_2(A_{Y_i}(s)+1)}\right)\left(1-\frac{m}{\epsilon_1(L_{Y_i}(t)+1)-\epsilon_2A_{Y_i}(s)}\right)\\
&\ \ \ \ \ \ \prod_{1\le i\ne j\le N}\prod_{s\in Y_i}\left(1-\frac{m}{a_i-a_j-\epsilon_1L_{Y_j}(s)+\epsilon_2(A_{Y_i}(s)+1)}\right)\prod_{t\in Y_j}\left(1-\frac{m}{-a_j+a_i+\epsilon_1(L_{Y_i}(t)+1)-\epsilon_2A_{Y_j}(t)}\right)\bigg|\le\\
&\sum_{k\ge 0}|q|^k\sum_{|\vec{Y}|=k}\prod_{i=1}^N\prod_{s\in Y_i}\left(1+\frac{\frac{|m^2-m(\epsilon_1+\epsilon_2)|}{|\epsilon_1|^2\gamma^2}}{h_{Y_1}(s)^2}\right)\prod_{i=1}^N\prod_{j\ne i}\left(1+\frac{|m|}{D_{ij}(\vec{a},\epsilon_1,\epsilon_2)}\right)^{2|Y_i|}\le\\
&\sum_{k\ge 0}\left[|q|\left(1+\frac{|m|}{D(\vec{a},\epsilon_1,\epsilon_2)}\right)^{2(N-1)}\right]^k\sum_{|\vec{Y}|=k}\prod_{i=1}^N\prod_{s\in Y_i}\left(1+\frac{\frac{|m^2-m(\epsilon_1+\epsilon_2)|}{|\epsilon_1|^2\gamma^2}}{h_{Y_1}(s)^2}\right)=\\
&\sum_{Y_1,\dots,Y_N\in\mathbb{Y}}\left[|q|\left(1+\frac{|m|}{D(\vec{a},\epsilon_1,\epsilon_2)}\right)^{2(N-1)}\right]^{\sum_{i=1}^N|Y_i|}\prod_{i=1}^N\prod_{s\in Y_i}\left(1+\frac{\frac{|m^2-m(\epsilon_1+\epsilon_2)|}{|\epsilon_1|^2\gamma^2}}{h_{Y_1}(s)^2}\right)=\\
&\left\{\sum_{Y\in\mathbb{Y}}\left[|q|\left(1+\frac{|m|}{D(\vec{a},\epsilon_1,\epsilon_2)}\right)^{2(N-1)}\right]^{|Y|}\prod_{s\in Y}\left(1+\frac{\frac{|m^2-m(\epsilon_1+\epsilon_2)|}{|\epsilon_1|^2\gamma^2}}{h_{Y}(s)^2}\right)\right\}^N=\\
&\phi\left(|q|\left(1+\frac{|m|}{D(\vec{a},\epsilon_1,\epsilon_2)}\right)^{2(N-1)}\right)^{N\left(-\frac{|m^2-m(\epsilon_1+\epsilon_2)|}{|\epsilon_1|^2\gamma^2}-1\right)},
\end{aligned}
\end{equation}
where in the last line we used \eqref{teo1.2} with $x=|q|\left(1+\frac{|m|}{D(\vec{a},\epsilon_1,\epsilon_2)}\right)^{2(N-1)}$ and $z=-\frac{|m^2-m(\epsilon_1+\epsilon_2)|}{|\epsilon_1|^2\gamma^2}$. Hence, as in the previous case, the instanton partition function is convergent in the region defined by
\begin{equation}\label{noncelha}
|q|< \left(1+\frac{|m|}{D(\vec{a},\epsilon_1,\epsilon_2)}\right)^{-2(N-1)}.
\end{equation}

\begin{remark} 
Let us note that in the case $\epsilon_2/\epsilon_1\in\mathbb{R}_{<0}$ the 2-dimensional lattice $\Lambda(\epsilon_1,\epsilon_2)$ degenerates into a 1-dimensional lattice. Therefore, if we move sufficiently away from the line spanned by $\epsilon_1$ in the complex plane, that is, if, for every $i\ne j$, $a_i-a_j$ has a big enough distance from the set $\{z\in\mathbb{C}\ |\ z=r\epsilon_1, r\in\mathbb{R}\}$, the constant $D(\vec{a},\epsilon_1,\epsilon_2)$ can become very large and the radius of convergence tends to 1.
\end{remark}

\subsection{Corollary \ref{corolN=2}: from $\mathcal{N}=2^*$ to $\mathcal{N}=2$ SYM}

The results on the convergence of the $\mathcal{N}=2$ instanton partition function can be deduced from the ones on the $\mathcal{N}=2^*$ instanton partition function. Indeed, if one considers the double scaling limit in which the mass of the adjoint multiplet $m$ becomes large $m\to\infty$ and the instanton parameter $q$ becomes small, $q\to 0$, in such a way that $\Lambda:=qm^{2N}$ remains finite, the instanton partition function of the $\mathcal{N}=2^*$ $U(N)$ theory \eqref{N=2*UN} reduces to \eqref{N=2UN} in the expansion parameter $\Lambda$ instead of $q$. 

From \eqref{noncelha}, we find
\begin{equation}
|m|^{2N}|q|\leq \left(1+\frac{|m|}{D(\vec{a},\epsilon_1,\epsilon_2)}\right)^{2}
\left(\frac{1}{|m|}+\frac{1}{D(\vec{a},\epsilon_1,\epsilon_2)}\right)^{-2N}
\end{equation}
which in the above limit reduces to 
$|\Lambda|<\infty$.

\subsection{Corollary \ref{corolN=4}: from $\mathcal{N}=2^*$ to $\mathcal{N}=4$}

The instanton partition function of the $\mathcal{N}=4$ $U(N)$ gauge theory can be written as
\begin{equation}
\begin{aligned}
Z_{\mathrm{inst}}^{\mathcal{N}=4, U(N)}=\sum_{k\ge 0}q^k\sum_{|Y|=k}1=\sum_{k\ge 0}q^kp_N(k)=\prod_{j=1}^{\infty}\frac{1}{(1-q^j)^N}=\phi(q)^{-N},
\end{aligned}
\end{equation}
which is convergent in the region $|q|< 1$.

This result can also be obtained from the analysis of the $\mathcal{N}=2^*$ $U(N)$ theory setting to zero the mass of the adjoint multiplet, as it is obvious from \eqref{noncelha}.



\section{$U(N)$ Instanton Partition Functions with Fundamental Matter}

Also in this case we will work under the assumptions \eqref{conditions}. Moreover, we assume $\epsilon_1+\epsilon_2=0$ and set
\begin{equation}\label{mu}
\epsilon:=\epsilon_1=-\epsilon_2, \quad \alpha_i:=a_i/\epsilon, \quad \mu_r:=m_r/\epsilon. 
\end{equation}
In this notation, the instanton partition function reads
\begin{equation}\label{functionfundamental}
\begin{aligned}
Z_{\mathrm{inst}}^{U(N)\ , N_f}=\sum_{k\ge 0}\left(q\epsilon^{N_f-2N}\right)^k\sum_{|\vec{Y}|=k}\prod_{i,j=1}^N&\prod_{(m,n)\in Y_i}\frac{1}{\alpha_i-\alpha_j-h_{Y_i}((m,n))+(Y'_i)_m-(Y'_j)_m}\\
&\prod_{(m,n)\in Y_j}\frac{1}{\alpha_i-\alpha_j+h_{Y_j}((m,n))-(Y'_j)_m+(Y'_i)_m}\\
\prod_{i=1}^N&\prod_{(m,n)\in Y_i}\prod_{r=1}^{N_f}\left[\alpha_i+\mu_r+m-n\right].
\end{aligned}
\end{equation}
We observe that in this case \eqref{conditions} reduces to
$\alpha_i-\alpha_j\notin\mathbb{Z}$ for every $1\le i<j\le N$. 

The main result we find is
\begin{teo}\label{teoremafondamentale}
The instanton partition function of the $U(N)$ gauge theory with $N_f=2N$ fundamental multiplets has at least a finite radius of convergence.

The instanton partition function of the $U(N)$ gauge theory with $N_f<2N$ fundamental multiplets is absolutely convergent over the whole complex plane.

\end{teo} 

We will consider in \eqref{functionfundamental} the sum starting from $k\ge 1$, as it does not change the convergence properties of the series.

There will be many steps necessary to arrive at our final result, so it will be useful to divide the coefficient functions into simpler factors and analyze them separately. 

We start by considering the products over the boxes of one specific Young diagram, let us take $Y_1$, which are
\begin{equation}\label{1}
\prod_{(m,n)\in Y_1}\frac{\prod_{r=1}^{N_f}\left[\alpha_1+\mu_r+m-n\right]}{h_{Y_1}((m,n))^2}\prod_{j\ne 1}\prod_{(m,n)\in Y_1}\frac{1}{(\alpha_1-\alpha_j-h_{Y_1}((m,n))+(Y'_1)_m-(Y'_j)_m)^2}.
\end{equation}
We first analyze the $N_f=2N$ case of the theorem, in which we have the same number of factors in the numerator and denominator of \eqref{1}. In particular, we can factor \eqref{1} in two types of products:
\begin{equation}
\prod_{(m,n)\in Y_1}\frac{\alpha_1+\mu_r+m-n}{h_{Y_1}((m,n))}\ \ \text{with}\ \ r\in\{1,\dots,N_f\},
\end{equation}
and
\begin{equation}\label{2}
\prod_{(m,n)\in Y_1}\frac{\alpha_1+\mu_r+m-n}{\alpha_1-\alpha_j-h_{Y_1}((m,n))+(Y'_1)_m-(Y'_j)_m}\ \ \text{with}\ \ r\in\{1,\dots,N_f\}\ \text{and}\ j\in\{2,\dots,N\}.
\end{equation}
The key result on the first kind of product is the following
\begin{lemma}\label{lemmaU1}
For every Young diagram $Y$ with $k\ge 1$ boxes and for every fixed complex number $z$, the following inequality holds
\begin{equation}\label{chiave1}
\begin{aligned}
\prod_{(i,j)\in Y}\bigg|\frac{z+i-j}{h_Y((i,j))}\bigg|<&\sqrt{\frac{k+2\max\{1,|z|\}\sqrt{k}-1}{2\pi k(2\max\{1,|z|\}\sqrt{k}-1)}}\left(1+\frac{2\sqrt{k}\max\{1,|z|\}-1}{k}\right)^k\left(1+\frac{k}{2\sqrt{k}\max\{1,|z|\}-1}\right)^{2\max\{1,|z|\}\sqrt{k}-1}\times\\
&\times\exp(\frac{1}{12(k+2\sqrt{k}\max\{1,|z|\}-1)}-\frac{1}{12k+1}-\frac{1}{12(2\sqrt{k}\max\{1,|z|\}-1)+1}).
\end{aligned}
\end{equation}
\end{lemma}
We will denote $f(z,k)$ the function on the right hand side of \eqref{chiave1}.

The key result on the second kind of product is the following
\begin{lemma}\label{lemmaU2}
For every pair of diagrams $(Y_1,Y_2)$ with $|Y_1|+|Y_2|=k\ge 1$ and for every pair of fixed complex numbers $z_1,z_2$, the following inequality holds
\begin{equation}\label{boundU2}
\begin{aligned}
&\prod_{(i,j)\in Y_1}\bigg|\frac{z_1+i-j}{\alpha_1-\alpha_2-h_{Y_1}((i,j))+(Y'_1)_i-(Y'_2)_i}\bigg|\prod_{(i,j)\in Y_2}\bigg|\frac{z_2+i-j}{\alpha_1-\alpha_2+h_{Y_2}((i,j))-(Y'_2)_i+(Y'_1)_i}\bigg|\le\\
&\left(\frac{16}{\min\{1,|\alpha_1-\alpha_2|\}}\left(1+\frac{|\alpha_1-\alpha_2|}{C_{12}(\vec{\alpha})}\right)\right)^{k}f(z_1,k)f(z_2,k),
\end{aligned}
\end{equation}
where
\begin{equation}
C_{ij}(\vec{\alpha})=\min_{n\in\mathbb{Z}}|\alpha_i-\alpha_j-n|>0.
\end{equation}
\end{lemma}
We will use the notation 
\begin{equation}
g_{ij}(\vec{\alpha})=\frac{16}{\min\{1,|\alpha_i-\alpha_j|\}}\left( 1+\frac{|\alpha_i-\alpha_j|}{C_{ij}(\vec{\alpha})}\right).
\end{equation} 

\subsection{Proofs of Lemma \ref{lemmaU1} and Lemma \ref{lemmaU2}}

We start by proving Lemma \ref{lemmaU1}. The following results will be important in the proof. First, the following formula, that appears in equation 7.207 of exercise 7.50 of \cite{stanley}, is crucial:
\begin{prop}
For a Young diagram $Y$ with $k$ boxes, if $c(\sigma)$ denotes the number of cycles in the permutation $\sigma\in S_k$, and $\chi^Y(\sigma)$ is the character of the irreducible representation of $S_k$ associated to the partition $Y$ of $k$ and evaluated in the element $\sigma\in S_k$, then
\begin{equation}\label{polynomialidentity}
\prod_{(i,j)\in Y}\frac{z+i-j}{h_Y((i,j))}=\frac{1}{k!}\sum_{\sigma\in S_k}\chi^Y(\sigma)z^{c(\sigma)}.
\end{equation}
\end{prop}
This holds as a polynomial identity for all $z\in\mathbb{C}$.

Moreover, we will use also the following result, that is Lemma 5 of \cite{muller}:
\begin{prop}
Let $Y$ be a partition of $k\ge 1$. Let $\mathrm{sq}(Y)$ be the side length of the largest square contained in $Y$; that is, the largest $j$ such that $Y_j \ge j$. Let $\sigma\in S_k$ be a permutation with $c(\sigma)$ cycles. Then
\begin{equation}\label{lemmamuller}
|\chi^Y(\sigma)|\le (2\mathrm{sq}(Y))^{c(\sigma)}.
\end{equation}
\end{prop}
Finally, from \cite{tao,ford} and references therein, the following holds
\begin{prop}\label{proptao}
For every natural number $m$, the expectation value of $m^{c(\sigma)}$, over all the permutations of $S_k$ weighted uniformly, is equal to 
\begin{equation}
E(m^c)=\binom{k+m-1}{k}.
\end{equation}
\end{prop}
This result can be extended to noninteger $m$ by considering the right hand side of the equation as the generalized binomial coefficient.

{\it Proof of Lemma \ref{lemmaU1}.}
From the identity \eqref{polynomialidentity}, it follows that
\begin{equation}
\begin{aligned}
\prod_{(i,j)\in Y}\bigg|\frac{z+i-j}{h_Y((i,j))}\bigg|=\bigg|\frac{1}{k!}\sum_{\sigma\in S_k}\chi^Y(\sigma)z^{c(\sigma)}\bigg|\le\frac{1}{k!}\sum_{\sigma\in S_k}|\chi^Y(\sigma)|\cdot |z|^{c(\sigma)}.
\end{aligned}
\end{equation}
Moreover, using \eqref{lemmamuller} and the fact that one has always
\begin{equation}
\mathrm{sq}(Y)\le\sqrt{k},
\end{equation}
we can conclude  
\begin{equation}
|\chi^Y(\sigma)|\le (2\sqrt{k})^{c(\sigma)},
\end{equation}
so that
\begin{equation}
\bigg|\prod_{(i,j)\in Y}\frac{z+i-j}{h_Y((i,j))}\bigg|\le\frac{1}{k!}\sum_{\sigma\in S_k}[2\sqrt{k}|z|]^{c(\sigma)}\le \frac{1}{k!}\sum_{\sigma\in S_k}[2\sqrt{k}\max\{1,|z|\}]^{c(\sigma)}.
\end{equation}
Now, the expression 
\begin{equation}
\frac{1}{k!}\sum_{\sigma\in S_k}[2\sqrt{k}\max\{1,|z|\}]^{c(\sigma)}
\end{equation}
is the expectation value of $[2\sqrt{k}\max\{1,|z|\}]^{c(\sigma)}$ with the uniform measure, where all permutations have the same probability, given by $1/k!$. 
From Proposition \ref{proptao}, we can use the generalized binomial coefficient in order to obtain 
\begin{equation}
\bigg|\prod_{(i,j)\in Y}\frac{z+i-j}{h_Y((i,j))}\bigg|\le\binom{k+2\sqrt{k}\max\{1,|z|\}-1}{k}.
\end{equation}

We can write
\begin{equation}\label{binomialcrescente}
\binom{k+2\sqrt{k}\max\{1,|z|\}-1}{k}=\frac{\Gamma(k+2\sqrt{k}\max\{1,|z|\})}{\Gamma(k+1)\Gamma(2\sqrt{k}\max\{1,|z|\})}.
\end{equation}
Using Stirling approximation in the form
\begin{equation}
\sqrt {2\pi n}\left(\frac {n}{e}\right)^{n}e^{\frac {1}{12n+1}}<\Gamma(n+1)<\sqrt {2\pi n} \left(\frac {n}{e}\right)^{n}e^{\frac {1}{12n}}, 
\end{equation}
we have 
\begin{equation}
\begin{aligned}
&\binom{k+2\sqrt{k}\max\{1,|z|\}-1}{k}<\\
&\frac{\sqrt{2\pi(k+2\max\{1,|z|\}\sqrt{k}-1)}}{\sqrt{2\pi k}\sqrt{2\pi(2\max\{1,|z|\}\sqrt{k}-1)}}\frac{(k+2\max\{1,|z|\}\sqrt{k}-1)^{k+2\max\{1,|z|\}\sqrt{k}-1}e^ke^{2\max\{1,|z|\}\sqrt{k}-1}}{k^k(2\max\{1,|z|\}\sqrt{k}-1)^{2\max\{1,|z|\}\sqrt{k}-1}e^{k+2\max\{1,|z|\}\sqrt{k}-1}}\times\\
&\times\exp(\frac{1}{12(k+2\sqrt{k}\max\{1,|z|\}-1)}-\frac{1}{12k+1}-\frac{1}{12(2\sqrt{k}\max\{1,|z|\}-1)+1})=\\
=&\sqrt{\frac{k+2\max\{1,|z|\}\sqrt{k}-1}{2\pi k(2\max\{1,|z|\}\sqrt{k}-1)}}\left(1+\frac{2\sqrt{k}\max\{1,|z|\}-1}{k}\right)^k\left(1+\frac{k}{2\sqrt{k}\max\{1,|z|\}-1}\right)^{2\max\{1,|z|\}\sqrt{k}-1}\times\\
&\times\exp(\frac{1}{12(k+2\sqrt{k}\max\{1,|z|\}-1)}-\frac{1}{12k+1}-\frac{1}{12(2\sqrt{k}\max\{1,|z|\}-1)+1}).
\end{aligned}
\end{equation}
\hfill $\square$

\begin{remark}
Let us remark that the binomial coefficient \eqref{binomialcrescente} is increasing in $k$. Indeed, considering the ratio of the binomial coefficient with $k=r+1$ and $k=r$, we have that
\begin{equation}
\begin{aligned}
&\frac{\binom{r+1+2\sqrt{r+1}\max\{1,|z_l|\}-1}{r+1}}{\binom{ r+2\sqrt{r}\max\{1,|z_l|\}-1}{r}}\ge \frac{\binom{r+2\sqrt{r}\max\{1,|z_l|\}}{r+1}}{\binom{ r+2\sqrt{r}\max\{1,|z_l|\}-1}{r}}=\frac{r+2\sqrt{r}\max\{1,|z_l|\}}{r+1}\ge 1.
\end{aligned}
\end{equation}
Therefore, when we consider a $N$-tuple of Young diagrams $Y_1,\dots,Y_N$ with $|Y_i|=k_i$ and $\sum_{i=1}^Nk_i=k\ge 1$, we can bound the quantity $\prod_{(m,n)\in Y_i}\frac{z+m-n}{h_{Y_i}((m,n))}$ with $f(z,k)$. This bound holds also if the diagram $Y_i$ is empty, since $f(z,k)>1$ if $k\ge 1$.
\end{remark}

In order to prove Lemma \ref{lemmaU2}, we need some further preliminary results which we now discuss.
Since we have found a sharp bound for the products of the form
\begin{equation}
\prod_{(i,j)\in Y_1}\frac{z+i-j}{h_{Y_1}((i,j))},
\end{equation}
we can write the second type of product \eqref{2} as (we fix $j=2$ in \eqref{2} for simplicity)
\begin{equation}
\prod_{(i,j)\in Y_1}\frac{z+i-j}{\alpha_1-\alpha_2-h_{Y_1}((i,j))+(Y'_1)_i-(Y'_2)_i}=\prod_{(i,j)\in Y_1}\frac{z+i-j}{h_{Y_1}((i,j))}\frac{h_{Y_1}((i,j))}{\alpha_1-\alpha_2-h_{Y_1}((i,j))+(Y'_1)_i-(Y'_2)_i},
\end{equation}
and so we can reduce to estimate the products of the form 
\begin{equation}
\prod_{(i,j)\in Y_1}\frac{h_{Y_1}((i,j))}{\alpha_1-\alpha_2-h_{Y_1}((i,j))+(Y'_1)_i-(Y'_2)_i}.
\end{equation}

Let us fix a pair of Young diagrams $Y_1,Y_2$ with $|Y_1|+|Y_2|=k\ge 1$. 
Let us consider first the product over the boxes of $Y_1$. We will suppose $Y_1$ to be nonempty, otherwise the product would clearly be bounded with 1, and the final estimate would also include that case. Let us divide the set of boxes of $Y_1$ in two subsets: we call $B_1(Y_1)$ the set of boxes of $Y_1$ for which $h_{Y_1}((i,j))=(Y'_1)_i-(Y'_2)_i$, and $B_2(Y_1)$ the set of boxes of $Y_1$ for which $h_{Y_1}((i,j))\ne (Y'_1)_i-(Y'_2)_i$. We have then
\begin{equation}\label{threeproducts}
\begin{aligned}
&\prod_{(i,j)\in Y_1}\frac{h_{Y_1}((i,j))}{\alpha_1-\alpha_2-h_{Y_1}((i,j))+(Y'_1)_i-(Y'_2)_i}=\\
&\prod_{(i,j)\in B_1(Y_1)}\frac{h_{Y_1}((i,j))}{\alpha_1-\alpha_2-h_{Y_1}((i,j))+(Y'_1)_i-(Y'_2)_i}\prod_{(i,j)\in B_2(Y_1)}\frac{h_{Y_1}((i,j))}{\alpha_1-\alpha_2-h_{Y_1}((i,j))+(Y'_1)_i-(Y'_2)_i}=\\
&\prod_{(i,j)\in B_1(Y_1)}\frac{h_{Y_1}((i,j))}{\alpha_1-\alpha_2}\prod_{(i,j)\in B_2(Y_1)}\frac{h_{Y_1}((i,j))}{\alpha_1-\alpha_2-h_{Y_1}((i,j))+(Y'_1)_i-(Y'_2)_i}\prod_{(i,j)\in B_2(Y_1)}\frac{-h_{Y_1}((i,j))+(Y'_1)_i-(Y'_2)_i}{-h_{Y_1}((i,j))+(Y'_1)_i-(Y'_2)_i}=\\
&\prod_{(i,j)\in B_1(Y_1)}\frac{(Y'_1)_i-(Y'_2)_i}{\alpha_1-\alpha_2}\prod_{(i,j)\in B_2(Y_1)}\frac{-h_{Y_1}((i,j))+(Y'_1)_i-(Y'_2)_i}{\alpha_1-\alpha_2-h_{Y_1}((i,j))+(Y'_1)_i-(Y'_2)_i}\prod_{(i,j)\in B_2(Y_1)}\frac{h_{Y_1}((i,j))}{-h_{Y_1}((i,j))+(Y'_1)_i-(Y'_2)_i}.
\end{aligned}
\end{equation}
We will consider the three products in the last line one by one.

\begin{lemma}\label{lemma3.7}
The first product in the last line of \eqref{threeproducts} can be bounded as follows
\begin{equation}\label{primopezzo}
\prod_{(i,j)\in B_1(Y_1)}\bigg|\frac{(Y'_1)_i-(Y'_2)_i}{\alpha_1-\alpha_2}\bigg|\le\frac{2^k}{\min\{1,|\alpha_1-\alpha_2|\}^k}.
\end{equation}
\end{lemma}

\begin{proof}
See Appendix \ref{b1}.
\end{proof}

\begin{remark}
Since for every fundamental hypermultiplet there is an identical product over the boxes in $B_1(Y_2)$, we notice that, for a given index $i$ of the box, only one of the two equalities $h_{Y_1}((i,j))=(Y'_1)_i-(Y'_2)_i$ and $h_{Y_2}((i,j))=(Y'_2)_i-(Y'_1)_i$ can be satisfied, since the left hand sides are always positive, but the right hand sides are one the opposite of the other. Therefore, for a fixed index $i$, only one of the factors $$\frac{(Y'_1)_i-(Y'_2)_i}{\alpha_1-\alpha_2}\ \ \text{and}\ \ \frac{(Y'_2)_i-(Y'_1)_i}{\alpha_1-\alpha_2}$$ appears in the product over the boxes in $B_1(Y_1)$ and over the boxes in $B_1(Y_2)$, so the previous estimate actually bounds the product of the two products of the first kind (the one for $Y_1$ and the one for $Y_2$).
\end{remark}

We now pass to the second product in \eqref{threeproducts}. 
\begin{lemma}
The second product in the last line of \eqref{threeproducts} can be bounded as follows
\begin{equation}\label{secondopezzo}
\prod_{(i,j)\in B_2(Y_1)}\bigg|\frac{-h_{Y_1}((i,j))+(Y'_1)_i-(Y'_2)_i}{\alpha_1-\alpha_2-h_{Y_1}((i,j))+(Y'_1)_i-(Y'_2)_i}\bigg|\le\biggl( 1+\frac{|\alpha_1-\alpha_2|}{C_{12}(\vec{\alpha})}\biggr)^{|Y_1|},
\end{equation}
where $$C_{12}(\vec{\alpha})=\min_{n\in\mathbb{Z}}|\alpha_1-\alpha_2-n|.$$
\end{lemma}
\begin{proof}
We have
\begin{equation}
\begin{aligned}
&\prod_{(i,j)\in B_2(Y_1)}\frac{-h_{Y_1}((i,j))+(Y'_1)_i-(Y'_2)_i}{\alpha_1-\alpha_2-h_{Y_1}((i,j))+(Y'_1)_i-(Y'_2)_i}=\prod_{(i,j)\in B_2(Y_1)}\frac{\alpha_1-\alpha_2-h_{Y_1}((i,j))+(Y'_1)_i-(Y'_2)_i-(\alpha_1-\alpha_2)}{\alpha_1-\alpha_2-h_{Y_1}((i,j))+(Y'_1)_i-(Y'_2)_i}=\\
&=\prod_{(i,j)\in B_2(Y_1)}\biggl(1-\frac{\alpha_1-\alpha_2}{\alpha_1-\alpha_2-h_{Y_1}((i,j))+(Y'_1)_i-(Y'_2)_i}\biggr),
\end{aligned}
\end{equation}
and so
\begin{equation}
\prod_{(i,j)\in B_2(Y_1)}\bigg|\frac{-h_{Y_1}((i,j))+(Y'_1)_i-(Y'_2)_i}{\alpha_1-\alpha_2-h_{Y_1}((i,j))+(Y'_1)_i-(Y'_2)_i}\bigg|\le\biggl( 1+\frac{|\alpha_1-\alpha_2|}{C_{12}(\vec{\alpha})}\biggr)^{|Y_1|},
\end{equation}
where $$C_{12}(\vec{\alpha})=\min_{n\in\mathbb{Z}}|\alpha_1-\alpha_2-n|.$$
\end{proof}

We finally bound the third product.
\begin{lemma}\label{lemma3.9}
The third product in the last line of \eqref{threeproducts} can be bounded as follows
\begin{equation}\label{terzopezzo}
\prod_{(i,j)\in B_2(Y_1)}\bigg|\frac{h_{Y_1}((i,j))}{-h_{Y_1}((i,j))+(Y'_1)_i-(Y'_2)_i}\bigg|\le 8^{|Y_1|}.
\end{equation}
\end{lemma}

\begin{proof}
See Appendix \ref{b2}.
\end{proof}

\begin{remark}
Referring to the proof in the appendix and considering the analogous product over the boxes in $Y_2$, we would have to bound the product
\begin{equation}
\prod_{(i,j)\in B_2(Y_2)\cap i\text{th\ row\ of\ }Y_2}\bigg|\frac{h_{Y_2}((i,j))}{h_{Y_2}((i,j))-[(Y'_2)_i-(Y'_1)_i]}\bigg|.
\end{equation}
But then, for a fixed $i$, we have either $(Y'_1)_i-(Y'_2)_i=0$, $(Y'_1)_i-(Y'_2)_i> 0$ or $(Y'_2)_i-(Y'_1)_i>0$. If we are in the first case, both products over the boxes in the $i$th row of $Y_1$ and over the $i$th row of $Y_2$ are bounded by 1. If we are in the second case, the product over the boxes in the $i$th row of $Y_2$ is bounded by 1, and, if we are in the third case,  the product over the boxes in the $i$th row of $Y_1$ is bounded by 1. Therefore, for every $i$, only one product has to be considered to give an upper bound. Hence, the previous bound, with $|Y_1|$ replaced by $k$, is a bound for the product of the two products of the third kind (the one for $Y_1$ and the one for $Y_2$).
\end{remark}

We are now finally ready to prove 
Lemma \ref{lemmaU2}.
\\

{\it Proof of Lemma \ref{lemmaU2}} :
Putting together \eqref{primopezzo}, \eqref{secondopezzo} and \eqref{terzopezzo}, and using the remarks after the previous lemmas, we conclude that
\begin{equation}\label{usareperquiver}
\begin{aligned}
&\prod_{(i,j)\in Y_1}\bigg|\frac{h_{Y_1}((i,j))}{\alpha_1-\alpha_2-h_{Y_1}((i,j))+(Y'_1)_i-(Y'_2)_i}\bigg|\prod_{(i,j)\in Y_2}\bigg|\frac{h_{Y_2}((i,j))}{\alpha_1-\alpha_2+h_{Y_2}((i,j))-(Y'_2)_i+(Y'_1)_i}\bigg|\le\\
&\le\left(\frac{16}{\min\{1,|\alpha_1-\alpha_2|\}}\left( 1+\frac{|\alpha_1-\alpha_2|}{C_{12}(\vec{\alpha})}\right)\right)^{k}.
\end{aligned}
\end{equation}
To conclude the proof of lemma \ref{lemmaU2}, it only remains to include the bounds of the products of the form analyzed in lemma \ref{lemmaU1}, both for $Y_1$ and $Y_2$.
Since the inequality
\begin{equation}
\prod_{(i,j)\in Y_l}\frac{z_l+i-j}{h_{Y_l}((i,j))}\le f(z_l,k)
\end{equation}
holds for both $l=1,2$, we can write the following estimate
\begin{equation}
\begin{aligned}
&\prod_{(i,j)\in Y_1}\bigg|\frac{z_1+i-j}{\alpha_1-\alpha_2-h_{Y_1}((i,j))+(Y'_1)_i-(Y'_2)_i}\bigg|\prod_{(i,j)\in Y_2}\bigg|\frac{z_2+i-j}{\alpha_1-\alpha_2+h_{Y_2}((i,j))-(Y'_2)_i+(Y'_1)_i}\bigg|\le\\
&\le\left(\frac{16}{\min\{1,|\alpha_1-\alpha_2|\}}\left( 1+\frac{|\alpha_1-\alpha_2|}{C_{12}(\vec{\alpha})}\right)\right)^{k}f(z_1,k)f(z_2,k).
\end{aligned}
\end{equation}
\hfill $\square$

\subsection{Proof of Theorem \ref{teoremafondamentale}}

In the $N_f=2N$ case, using Lemma \ref{lemmaU1} and Lemma \ref{lemmaU2}, we can arrange the products in the numerator and denominator of the coefficients of the instanton partition function \eqref{functionfundamental} to conclude that
\begin{equation}\label{Nf=2N}
|Z_{\mathrm{inst}}^{U(N)\ N_f=2N}|\le\sum_{k\ge 1}|q|^kp_N(k)\left[\prod_{i=1}^N\prod_{r=1}^{N_f=2N}f(\alpha_i+\mu_r,k)\right]\prod_{\{(i,j)\in\{1,\dots,N\}^2\ |\ i\ne j\}}g_{ij}(\vec{\alpha})^k,
\end{equation}
where $p_N(k)$ denotes the number of $N$-coloured partitions of the integer $k$. If $p(k)$ denotes the number of partitions of $k$, we can bound $p_N(k)$ with $p(k)^{N+1}$,
since the former can be seen as the number of partitions of $N$ integers whose sum equals $k$, and so any of these $N$ numbers has to be smaller than $k$. Moreover, we can use the following estimate, known as Ramanujan-Hardy formula \cite{ramanujan}:
\begin{prop} 
If $p(k)$ is the number of partitions of the natural number $k$, the following holds:
\begin{equation}
p(k)\sim\frac{1}{4\sqrt{3}k}\exp(\pi\sqrt{\frac{2k}{3}}), \ \ \text{for\ } k\to\infty.
\end{equation}
\end{prop}
Therefore, applying the root test to \eqref{Nf=2N}, and using that
\begin{equation}
\begin{aligned}
&\lim_{k\to\infty}p(k)^{1/k}=\lim_{k\to\infty}\left(\frac{1}{4\sqrt{3}k}\exp(\pi\sqrt{\frac{2k}{3}})\right)^{1/k}=1\\
&\lim_{k\to\infty}\left(f(\alpha_i+\mu_r,k)\right)^{1/k}=1\ \ \forall i=1,\dots,N\ \forall r=1,\dots,N_f,
\end{aligned}
\end{equation}
we conclude that the radius of convergence of the right hand side is given by
\begin{equation}
\prod_{\{(i,j)\in\{1,\dots,N\}^2\ |\ i\ne j\}}\left[g_{ij}(\vec{\alpha})\right]^{-1}=\prod_{\{(i,j)\in\{1,\dots,N\}^2\ |\ i\ne j\}}\left[\frac{16}{\min\{1,|\alpha_i-\alpha_j|\}}\left( 1+\frac{|\alpha_i-\alpha_j|}{C_{ij}(\vec{\alpha})}\right)\right]^{-1}.
\end{equation}
Therefore, we can conclude the first part of the theorem, that is the fact that the instanton partition function of the $U(N)$ gauge theory with $N_f=2N$ fundamental multiplets with the Omega background $\epsilon_1+\epsilon_2=0$ is absolutely convergent at least for
\begin{equation}\label{radiusfund}
|q|<\prod_{\{(i,j)\in\{1,\dots,N\}^2\ |\ i\ne j\}}\left[\frac{16}{\min\{1,|\alpha_i-\alpha_j|\}}\left( 1+\frac{|\alpha_i-\alpha_j|}{C_{ij}(\vec{\alpha})}\right)\right]^{-1}.
\end{equation}
The case $N_f<2N$ can now be easily proved by noticing that the 
 decoupling limit of fundamental hypermultiplets is achieved with the double scaling limit in which $q\to 0$ and one of the masses, say $m_1$, goes to infinity $m_1\to \infty$, in such a way that $\tilde{\Lambda}=qm_1$ remains finite. Indeed, from the expression \eqref{functionfundamentalperdecoupling}, one can see that in this limit the function becomes
\begin{equation}
\begin{aligned}
Z_{\mathrm{inst}}^{U(N),\  N_f}=\sum_{k\ge 0}\left(qm_1\right)^k\sum_{|\vec{Y}|=k}\prod_{i,j=1}^N&\prod_{(m,n)\in Y_i}\frac{1}{a_i-a_j-\epsilon_1L_{Y_j}((m,n))+\epsilon_2\left(A_{Y_i}((m,n))+1\right)}\\
&\prod_{(m,n)\in Y_j}\frac{1}{a_i-a_j+\epsilon_1\left(L_{Y_i}((m,n))+1\right)-\epsilon_2A_{Y_j}((m,n))}\\
\prod_{i=1}^N&\prod_{(m,n)\in Y_i}\left[1+\frac{a_i+\epsilon_1(m-1)+\epsilon_2(n-1)}{m_1}\right]\\
\prod_{i=1}^N&\prod_{(m,n)\in Y_i}\prod_{r=2}^{N_f}\left[a_i+\epsilon_1(m-1)+\epsilon_2(n-1)+m_r\right]\to\\
\sum_{k\ge 0}\tilde{\Lambda}^k\sum_{|\vec{Y}|=k}\prod_{i,j=1}^N&\prod_{(m,n)\in Y_i}\frac{1}{a_i-a_j-\epsilon_1L_{Y_j}((m,n))+\epsilon_2\left(A_{Y_i}((m,n))+1\right)}\\
&\prod_{(m,n)\in Y_j}\frac{1}{a_i-a_j+\epsilon_1\left(L_{Y_i}((m,n))+1\right)-\epsilon_2A_{Y_j}((m,n))}\\
\prod_{i=1}^N&\prod_{(m,n)\in Y_i}\prod_{r=2}^{N_f}\left[a_i+\epsilon_1(m-1)+\epsilon_2(n-1)+m_r\right],
\end{aligned}
\end{equation}
which is the instanton partition function with one fundamental hypermultiplet less. The radius of convergence of this latter series in $\tilde\Lambda$ can be obtained by multiplying \eqref{radiusfund} by $m_1$ and letting $m_1\to\infty$, which means that the series is absolutely convergent for any $\tilde\Lambda$.
The proof for lower $N_f$ is obtained by repeated application of the above argument.
\hfill$\square$

\section{On the convergence of Painlev\'e $\tau$-functions}

The Kiev formula conjectured in \cite{Gamayun:2013auu}
states that Painlev\'e $\tau$-functions can be expressed as discrete Fourier transforms of 
suitable full Nekrasov partition functions. This is the core issue of Painlev\'e/gauge theory correspondence \cite{Bonelli:2016qwg}.
Concretely, according to the Kiev formula, the PVI $\tau$-function is related to the Nekrasov function 
as follows
\begin{equation}\label{kiev}
\tau^{\rm VI}(q;\alpha,s)=q^{-\theta_0^2-\theta_t^2}(1-q)^{\theta_1\theta_t}\sum_{n\in{\mathbb Z}}s^n q^{(\alpha+n)^2}
Z_{\rm 1loop}^{U(2)\ N_f=4}(\alpha+n)Z_{\rm inst}^{U(2)\ N_f=4}(q,\alpha+n),
\end{equation}
where 
\begin{equation}\label{1loop}
Z_{\rm 1loop}^{U(2)\ N_f=4}(\alpha) =\frac{\prod_{\sigma,\sigma'=\pm}G(1+\theta_t+\sigma \theta_0+\sigma'(\alpha+n))G(1+\theta_1+\sigma \theta_{\infty}+\sigma'(\alpha+n))}{G(1+2(\alpha+n))G(1-2(\alpha+n))}
\end{equation}
is the one loop contribution to the full partition function written in terms of Barnes $G$ functions, and the re-scaled masses \eqref{mu}
are related to the $\theta$-parameters by 
$$\mu_1=\theta_1-\theta_\infty,\quad \mu_2=\theta_0-\theta_t,\quad\mu_3=\theta_0+\theta_t,\quad\mu_4=\theta_1+\theta_\infty. $$
The $\tau$-function \eqref{kiev} is the one associated to the isomonodromic deformation problem for the Riemann sphere with four regular singularities, with $\theta$s parameterizing the associated monodromies. 

In order to study the convergence properties of the series
\eqref{kiev}, we can make use of the results obtained in the previous Section 3 together with the asymptotic behaviour of the one-loop coefficients. 
The latter can be determined from the 
reflection formula:
\begin{equation}
G(1-z)=\frac{G(1+z)}{(2\pi)^{z}}\exp(\int_0^z\pi z'\cot(\pi z')\mathrm{d}z')
\end{equation}
and the asymptotic formula 
for $z\to\infty$ 
\cite{choi}
\begin{equation}
\log(G(1+a+z))=\frac{z+a}{2}\log(2\pi)
+\zeta'(-1)
-\frac{3z^2}{4}-az+\left(\frac{z^2}{2}-\frac{1}{12}+\frac{a^2}{2}+az\right)\log(z)+\mathcal{O}\left(\frac{1}{z}\right),
\end{equation}
which holds for all $a\in\mathbb{C}$ and where 
$\zeta'(-1)$ is a known $\zeta$-constant.
From this, we have that, for $a\in\mathbb{C}$ and $\mathbb{Z}\ni n\to\infty$,
\begin{equation}
\log(G(1+a+n))=\frac{n^2}{2}\log(n)-\frac{3n^2}{4}+\mathcal{O}(n\log(n)).
\end{equation}
To evaluate the $n\to \infty$ limit of the other set of Barnes functions, we note that the integral in the reflection formula is given by
\begin{equation}\label{integralasympt}
\int_0^z\pi z'\cot(\pi z')\mathrm{d}z'=\frac{\pi z\log(1-\exp(2\pi i z))-\frac{i}{2}\left(\pi^2z^2+\mathrm{Li}_2(\exp(2\pi i z))\right)}{\pi}.
\end{equation}
Since the asymptotic of the above integral is given by $-\frac{i}{2}\pi n^2+\mathcal{O}(n)$, 
we have that, 
for every $b\in\mathbb{C}$ and for $\mathbb{Z}\ni n\to\infty$,
\begin{equation}
\log(G(1-b-n))=\frac{n^2}{2}\log(n)-\frac{3n^2}{4}-\frac{i}{2}\pi n^2+\mathcal{O}(n\log(n)).
\end{equation}
Therefore, neglecting terms of order $n\log(n)$, which are subleading, the one-loop coefficient in the limit $\mathbb{Z}\ni n\to\infty$ reads
\begin{equation}\label{pippo}
\begin{aligned}
&\frac{\prod_{\sigma,\sigma'=\pm}G(1+\theta_t+\sigma \theta_0+\sigma'(\alpha+n))G(1+\theta_1+\sigma \theta_{\infty}+\sigma'(\alpha+n))}{G(1+2(\alpha+n))G(1-2(\alpha+n))}\quad \to\\ \to \quad 
&\frac{\left(n^{\frac{n^2}{2}}\exp(-\frac{3n^2}{4})\right)^8\left(\exp(-\frac{i\pi n^2}{2})\right)^4}{\left((2n)^{\frac{(2n)^2}{2}}\exp(-\frac{3(2n)^2}{4})\right)^2\left(\exp(-\frac{i\pi (2n)^2}{2})\right)}=\frac{1}{2^{4n^2}}.
\end{aligned}
\end{equation}
This immediately implies that the convergence radius of the $\tau^{\rm VI}$-function series is driven by the one of the $Z_{\rm inst}$ coefficient, for which we derived the lower bound \eqref{radiusfund} in Theorem \ref{teoremafondamentale}. 
Actually, as already mentioned in the Introduction, one expects from modularity that the true radius
of convergence is $|q|<1$.

The $\tau$-functions for Painlev\'e V and III${}_i$ $i=1,2,3$ equations are obtained by implementing in the gauge theory the suitable coalescence limits. These correspond to the holomorphic decoupling of fundamental masses,  already discussed in the previous Section 3 for the instanton sector. 
As far as the one-loop coefficient is concerned, the holomorphic decoupling lowers the number of factors in the numerator of \eqref{1loop}, which implies even stronger convergence properties driven by the denominator, as one can see from \eqref{pippo}. We therefore conclude that the corresponding Painlev\'e $\tau$-functions have an infinite radius of convergence. Actually, this was already shown
to hold for the PIII$_3$ equation in \cite{Its:2014lga}.

The above, together with Theorem \ref{teoremafondamentale},
provide a proof of the following
\begin{teo}
Let $2\alpha\notin \mathbb{Z}$. The $\tau$-function for 
PVI equation has at least a finite radius of absolute and uniform convergence, while those of 
PV and PIII$_i$ $i=1,2,3$ equations have an infinite radius of absolute and uniform convergence.
\end{teo}

Let us also mention that an extension of Kiev formula for the isomonodromic deformation problem on the torus was introduced in \cite{Bonelli:2019boe,Bonelli:2019yjd}. For the one-punctured torus the corresponding equations are given by Manin's elliptic form of PVI equation with specific values of the monodromy parameters, and the related $\tau$-function is obtained in terms of the partition function of the $U(2)$ $\mathcal{N}=2^*$ theory
\begin{equation}\label{kiev*}
\tau^{U(2)\ {\mathcal N}=2^*}(q;\alpha,s)=Z_D/Z_{\rm twist},
\end{equation}
where
$$
Z_{\rm twist}= 
q^{\alpha^2} \eta(q)^{-2} \theta_1(\alpha \tau+ \rho + Q(\tau))
\theta_1(\alpha \tau+ \rho - Q(\tau))
$$
is given in terms of the solution of the corresponding Painlev\'e equation $Q(\tau)$
and
\begin{equation}
Z_D=
\sum_{n\in{\mathbb Z}}s^n q^{(\alpha+n)^2}
Z_{\rm 1loop}^{U(2)\ {\mathcal N}=2^*}(\alpha+n)Z_{\rm inst}^{U(2)\ {\mathcal N}=2^*}(q,\alpha+n)
\end{equation}
with $s=e^{2\pi i \rho}$ and $q=e^{2\pi i \tau}$.
The one-loop coefficient is given by 
\begin{equation}
Z_{\rm 1loop}^{U(2)\ {\mathcal N}=2^*}= \frac{G(1-\mu-2(\alpha+n))G(1-\mu+2(\alpha+n))}{G(1+2(\alpha+n))G(1-2(\alpha+n))},
\end{equation}
where $\mu=m/\epsilon$ is the re-scaled adjoint mass.

\begin{teo}\label{tau2}
Let $2\alpha\notin \mathbb{Z}$. The $\tau$-function \eqref{kiev*} has at least a finite radius of absolute and uniform convergence.  
\end{teo}

\begin{proof}
With the same asymptotic formulas used before, see Appendix \ref{appe}, we have that, as $n\to\infty$,
\begin{equation}
\frac{G(1-\mu-2(\alpha+n))G(1-\mu+2(\alpha+n))}{G(1+2(\alpha+n))G(1-2(\alpha+n))}
\propto
(2n)^{\mu^2}
\left(\frac{\sin(\pi (\mu+2\alpha))}{\sin(2\pi \alpha)}\right)^{2n}
\end{equation}
up to $1/n$ corrections, where the proportionality constant is independent on $n$.
This does not get worst the convergence radius of the instanton sector and the proof follows from Theorem \ref{teoN=2}.
\end{proof}

\newpage
\appendix

\section{Conventions and Notations}\label{appendixA}

In this appendix we fix the conventions that we use in the main part of the work. We will mostly follow the notations of \cite{Bruzzo:2002xf, Alday:2009aq}.

\begin{defin}
A \emph{partition} of a positive integer $k$ is a finite non-increasing sequence of positive integers $Y_1\ge\dots\ge Y_r>0$ such that $\sum_{i=1}^rY_i=k$.
\end{defin}
We denote the number of partitions of $k$ as $p(k)$. The $Y_i$s that appear in a given partition are called \emph{parts} of the partition.
\begin{defin}
We say that a partition is \emph{$N$-coloured} if each part of the partition can have $N$ possible colours.
\end{defin}
We denote the number of $N$-coloured partitions of $k$ as $p_N(k)$.

We introduce some important functions related to the partitions of integers. Let $\tau$ be a complex number with $\mathrm{Im}\tau>0$, and let $q=e^{2\pi i\tau}$.
\begin{defin}
The \emph{Dedekind $\eta$ function} is defined as $$\eta(q)=q^{\frac{1}{24}}\prod_{n=1}^{\infty}(1-q^n).$$
\end{defin}
The requests on $\tau$ and $q$ are justified by the following:
\begin{prop}
The infinite product $$\prod_{n=1}^{\infty}(1-q^n)$$ converges absolutely if $|q|<1$.
\end{prop}
\begin{defin}
The \emph{Euler function} is defined as $$\phi(q)=\prod_{n=1}^{\infty}(1-q^n).$$
\end{defin}
Note that the Euler function coincides with the Dedekind $\eta$ function up to a factor $q^{\frac{1}{24}}$.
\begin{prop}\label{colouredpartition}
For every $N\ge 1$, the generating function for $p_N(k)$ is given by 
\begin{equation}
\sum_{k=0}^{\infty}p_N(k)q^k=\prod_{j=1}^{\infty}\frac{1}{(1-q^j)^N}.
\end{equation}
\end{prop}

In the text, we will always identify a partition of a natural number $k$ with a Young diagram $Y$ with $k$ boxes, arranged in left-justified rows, with the row lengths in non-increasing order, such that the parts $Y_1\ge Y_2\ge\dots\ge Y_r>0$ of $Y$ (such that $Y_1+\dots+Y_r=k$) denote the heights of the columns of the diagram.
Moreover, we will denote with $Y'_1\ge Y'_2\ge\dots\ge Y'_s>0$ the lengths of the rows of $Y$. We will denote with $\mathbb{Y}$ the set of all Young diagrams.

If every box $s$ is labeled with a pair of indices $(i,j)$, with $1\le i\le Y_j$ and $1\le j\le Y'_i$, that denotes its position in the diagram, we define the \emph{arm length} and the \emph{leg length} of $s$ as 
\begin{equation}
\begin{aligned}
A_Y(s)&=Y_j-i,\\
L_Y(s)&=Y'_i-j,
\end{aligned}
\end{equation}
respectively.

Moreover, we will use the following
\begin{defin}
If $Y$ is a Young diagram, and $s=(i,j)$ is one of its box, we call \emph{hook} of $s$ the set of boxes with indices $(a,b)$ such that $a=i$ and $b\ge j$ or $a\ge i$ and $b=j$.
\end{defin}
We denote with $h_Y((i,j))$ or $h_Y(s)$ the number of boxes in the hook of $s$ in $Y$. It is easy to see that, if $s\in Y$, then 
\begin{equation}
h_Y(s)=A_Y(s)+L_Y(s)+1.
\end{equation}

For a box $s=(i,j)$, we define the following quantities, crucial for the definitions of the instanton partition functions:
\begin{equation}
\begin{aligned}
E(a,Y_1,Y_2,s)&=a-\epsilon_1L_{Y_2}(s)+\epsilon_2(A_{Y_1}(s)+1)\\
\varphi\left(a,s=(i,j)\right)&=a+\epsilon_1(i-1)+\epsilon_2(j-1).
\end{aligned}
\end{equation}

We are now ready to define the useful contributions for the $U(N)$ instanton partition functions \cite{Flume:2002az, Bruzzo:2002xf}. We begin with the contribution of a bifundamental hypermultiplet of mass $m$:
\begin{equation}
\begin{aligned}
z_{\mathrm{bifund}}(\vec{a},\vec{Y};\vec{b},\vec{W};m)=\prod_{i,j=1}^N&\prod_{s\in Y_i}(E(a_i-b_j,Y_i,W_j,s)-m)\\ &\prod_{t\in W_j}(\epsilon_1+\epsilon_2-E(b_j-a_i,W_j,Y_i,t)-m),
\end{aligned}
\end{equation}
where with $\vec{Y}$ we denote an $N$-tuple $\vec{Y}=(Y_1,\dots,Y_N)$ of Young diagrams, and the same for $\vec{W}$, while $\vec{a}=(a_1,\dots,a_N)$ and $\vec{b}=(b_1,\dots,b_N)$ denote the vevs of the scalar component of the vector multiplets on the Coulomb branch.

From this, the contributions of an adjoint hypermultiplet of mass $m$ and of a vector multiplet can be written as
\begin{equation}
\begin{aligned}
z_{\mathrm{adj}}(\vec{a},\vec{Y},m)&=z_{\mathrm{bifund}}(\vec{a},\vec{Y},\vec{a},\vec{Y},m),\\
z_{\mathrm{vect}}(\vec{a},\vec{Y})&=[z_{\mathrm{adj}}(\vec{a},\vec{Y},0)]^{-1}.
\end{aligned}
\end{equation}

Finally, the contributions for fundamental and antifundamental hypermultiplets read as follows:
\begin{equation}
\begin{aligned}
z_{\mathrm{fund}}(\vec{a},\vec{Y},m)&=\prod_{i=1}^2\prod_{s\in Y_i}(\varphi(a_i,s)-m+\epsilon_1+\epsilon_2),\\
z_{\mathrm{antifund}}(\vec{a},\vec{Y},m)&=z_{\mathrm{fund}}(\vec{a},\vec{Y},\epsilon_1+\epsilon_2-m).
\end{aligned}
\end{equation}

We finally recall the expressions of the instanton partition functions analyzed in the text.

The instanton partition function of the $\mathcal{N}=2^*$ gauge theory with gauge group $U(N)$ can be written as
\begin{equation}\label{N=2*UN}
\begin{aligned}
Z_{\mathrm{inst}}^{\mathcal{N}=2^*, U(N)}=\sum_{k\ge 0}q^k\sum_{|\vec{Y}|=k}\prod_{i,j=1}^N&\prod_{s\in Y_i}\frac{a_i-a_j-\epsilon_1L_{Y_j}(s)+\epsilon_2(A_{Y_i}(s)+1)-m}{a_i-a_j-\epsilon_1L_{Y_j}(s)+\epsilon_2(A_{Y_i}(s)+1)}\\
&\prod_{t\in Y_j}\frac{-a_j+a_i+\epsilon_1(L_{Y_i}(t)+1)-\epsilon_2A_{Y_j}(t)-m}{-a_j+a_i+\epsilon_1(L_{Y_i}(t)+1)-\epsilon_2A_{Y_j}(t)},
\end{aligned}
\end{equation}
where the sum over $|\vec{Y}|=k$ means that we are summing over $N$-tuples of Young diagrams $(Y_1,\dots,Y_N)$ such that the sum of the number of the boxes in all the diagram is equal to $k$.

The instanton partition function of the $\mathcal{N}=2$ super Yang--Mills gauge theory with gauge group $U(N)$ can be written as
\begin{equation}\label{N=2UN}
\begin{aligned}
Z_{\mathrm{inst}}^{\mathcal{N}=2, U(N)}=\sum_{k\ge 0}q^k\sum_{|\vec{Y}|=k}\prod_{i,j=1}^N&\prod_{s\in Y_i}\frac{1}{a_i-a_j-\epsilon_1L_{Y_j}(s)+\epsilon_2(A_{Y_i}(s)+1)}\\
&\prod_{t\in Y_j}\frac{1}{-a_j+a_i+\epsilon_1(L_{Y_i}(t)+1)-\epsilon_2A_{Y_j}(s)}.
\end{aligned}
\end{equation}

For what concerns the instanton partition function of the $U(N)$ gauge theory with $N_f$ (anti)fundamental hypermultiplets, our analysis does not depend on whether the matter is in the fundamental or antifundamental representation, and in order to simplify the notation we will restrict to consider only the antifundamental matter.
Hence, we can write
\begin{equation}\label{functionfundamentalperdecoupling}
\begin{aligned}
Z_{\mathrm{inst}}^{\mathcal{N}=2\ U(N),\ N_f}=\sum_{k\ge 0}q^k\sum_{|\vec{Y}|=k}\prod_{i,j=1}^N&\prod_{(m,n)\in Y_i}\frac{1}{a_i-a_j-\epsilon_1L_{Y_j}((m,n))+\epsilon_2\left(A_{Y_i}((m,n))+1\right)}\\
&\prod_{(m,n)\in Y_j}\frac{1}{a_i-a_j+\epsilon_1\left(L_{Y_i}((m,n))+1\right)-\epsilon_2A_{Y_j}((m,n))}\\
\prod_{i=1}^N&\prod_{(m,n)\in Y_i}\prod_{r=1}^{N_f}\left[a_i+\epsilon_1(m-1)+\epsilon_2(n-1)+m_r\right],
\end{aligned}
\end{equation}
where $m_r$, $r=1,\dots,N_f$, are the masses of the antifundamental hypermultiplets.

\section{Proof of Lemma \ref{lemma3.7}}
\label{b1}

We know that the boxes in $B_1(Y_1)$ satisfy $h_{Y_1}((i,j))=(Y'_1)_i-(Y'_2)_i$. This can happen at most for one box in each row of $Y_1$, since the left hand side strictly decreases moving on the right on a fixed row of the diagram, while the right hand side remains constant. Therefore, we can bound the product as follows:
\begin{equation}\label{A}
\begin{aligned}
&\prod_{(i,j)\in B_1(Y_1)}\bigg|\frac{(Y'_1)_i-(Y'_2)_i}{\alpha_1-\alpha_2}\bigg|=\frac{\prod_{(i,j)\in B_1(Y_1)}|(Y'_1)_i-(Y'_2)_i|}{|\alpha_1-\alpha_2|^{|B_1(Y_1)|}}\le\\
&\frac{\max\{1,|(Y'_1)_1-(Y'_2)_1|\}\cdots \max\{1,|(Y'_1)_{(Y_1)_1}-(Y'_2)_{(Y_1)_1}|\}}{|\alpha_1-\alpha_2|^{|B_1(Y_1)|}},
\end{aligned}
\end{equation}
where we bounded the product in the numerator with the product of all the differences between rows' lenghts ($(Y_1)_1$ is the height of the first column of $Y_1$, that is the number of rows of $Y_1$), and we modified the factors taking the maximum with 1, because it could happen that, in a fixed row $i$ of $Y_1$, there is not a box which is in $B_1(Y_1)$ and $(Y'_1)_i=(Y'_2)_i$ holds, and we want to avoid that the right hand side vanishes for this reason.

From the last term in \eqref{A}, we can bound the numerator using the geometric-arithmetic mean inequality:
\begin{equation}\label{B}
\begin{aligned}
&\max\{1,|(Y'_1)_1-(Y'_2)_1|\}\cdots \max\{1,|(Y'_1)_{(Y_1)_1}-(Y'_2)_{(Y_1)_1}|\}\le\biggl(\frac{\max\{1,|(Y'_1)_1-(Y'_2)_1|\}+\dots+  \max\{1,|(Y'_1)_{(Y_1)_1}-(Y'_2)_{(Y_1)_1}|\}}{(Y_1)_1}\biggr)^{(Y_1)_1}\\
&\le\biggl(\frac{k}{(Y_1)_1}\biggr)^{(Y_1)_1}\le \binom{k}{(Y_1)_1}\le 2^k,
\end{aligned}
\end{equation}
where we used that for the binomial coefficient, for every $1\le k\le n$, the following bounds always hold $$\left(\frac{n}{k}\right)^k\le\binom{n}{k}<\left(\frac{n\cdot e}{k}\right)^k.$$
For the denominator, we have to distinguish the cases in which $|\alpha_1-\alpha_2|\ge 1$ and $|\alpha_1-\alpha_2|< 1$. In the first case, we simply bound the fraction with the bound of the numerator; in the second case, we have that, since $(Y_1)_1\le k$, $|\alpha_1-\alpha_2|^{|B_1(Y_1)|}\ge |\alpha_1-\alpha_2|^k$. Therefore, 
\begin{equation}
\prod_{(i,j)\in B_1(Y_1)}\bigg|\frac{(Y'_1)_i-(Y'_2)_i}{\alpha_1-\alpha_2}\bigg|\le\frac{2^k}{\min\{1,|\alpha_1-\alpha_2|\}^k}.
\end{equation}

\section{Proof of Lemma \ref{lemma3.9}}
\label{b2}

Let us first find a bound on the product over the boxes in $B_2(Y_1)$ in one fixed row of $Y_1$. After that, we will multiply the bounds on all the rows of $Y_1$.
We can write 
\begin{equation}
\begin{aligned}
\prod_{(i,j)\in B_2(Y_1)\cap i\text{th\ row\ of\ }Y_1}\bigg|\frac{h_{Y_1}((i,j))}{-h_{Y_1}((i,j))+(Y'_1)_i-(Y'_2)_i}\bigg|=\prod_{(i,j)\in B_2(Y_1)\cap i\text{th\ row\ of\ }Y_1}\bigg|\frac{h_{Y_1}((i,j))}{h_{Y_1}((i,j))-[(Y'_1)_i-(Y'_2)_i]}\bigg|.
\end{aligned}
\end{equation}
Note that the denominator is different from 0 for all the factors, since we are only multiplying over the boxes in $B_2(Y_1)$.

We will suppose $(Y'_1)_i> (Y'_2)_i$ for every $i$, since otherwise the previous product would be clearly bounded by 1 in the $i$th row.

Then, for a given row $i$, the product over the boxes in the $i$th row of $Y_1$ can be splitted in two parts: the product over the boxes for which $h_{Y_1}((i,j))-[(Y'_1)_i-(Y'_2)_i]$ is positive, and the product over the boxes for which the same quantity is negative. Note that, since we are assuming $(Y'_1)_i-(Y'_2)_i>0$, the latter product is present if and only if in the $i$th row there is a box, let us denote it with $(i,j^*)$, such that $h_{Y_1}((i,j^*))=[(Y'_1)_i-(Y'_2)_i]$, since the quantity $h_{Y_1}((i,j))-[(Y'_1)_i-(Y'_2)_i]$ is strictly decreasing moving to the right on a fixed row. 

Therefore, we first consider the product over the boxes for which that quantity is positive (that correspond to the boxes at the left of $(i,j^*)$ if this box is present in the $i$th row).
We can rewrite the factors of this first part of the product as
\begin{equation}
\frac{h_{Y_1}((i,j))}{h_{Y_1}((i,j))-[(Y'_1)_i-(Y'_2)_i]}=\frac{[(Y'_1)_i-(Y'_2)_i]+(Y'_2)_i-j+A_{Y_1}((i,j))+1}{(Y'_2)_i-j+A_{Y_1}((i,j))+1},
\end{equation}
which is of the form $$\prod_{j=1}^n\frac{a+b_j}{b_j},$$ with $b_j\in\mathbb{N}$ and $b_{j+1}<b_j$ for all $j$, and $a>0$ constant (since moving to the right $A_{Y_1}((i,j))$ decreases). But then a product of this form is bounded by
\begin{equation}
\frac{a+1}{1}\cdot\frac{a+2}{2}\cdots\frac{a+n}{n}.
\end{equation}
Indeed, if $a>0$ and $b>c>0$, it is always true that $$\frac{a+b}{b}\le\frac{a+c}{c},$$ since, under those hypothesis, $$\frac{a+b}{b}\le\frac{a+c}{c}\iff (a+b)c\le (a+c)b\iff ac\le ab\iff c\le b.$$ But then, $b_n\ge 1$ (since it is an integer number and for hypothesis it is positive), and, since $b_{j-1}>b_j$ for all $j=2,\dots,n$, we have that $b_j\ge n-j+1$ for all $j=1,\dots n-1$; so the previous bound holds. 

In our case, $n$ is at most $j^*-1$, so we can bound this first part of the product with
\begin{equation}\label{leftboxes}
\prod_{r=1}^{j^*-1}\frac{[(Y'_1)_i-(Y'_2)_i]+r}{r}=\binom{(Y'_1)_i-(Y'_2)_i+j^*-1}{j^*-1}.
\end{equation}

We can bound the second part of the product (if there are boxes on the right of $(i,j^*)$) as follows. First, from $h_{Y_1}((i,j^*))=(Y'_1)_i-(Y'_2)_i$, it follows that $$A_{Y_1}((i,j^*))+(Y'_2)_i+1=j^*.$$ Then, we rewrite
\begin{equation}
\begin{aligned}
h_{Y_1}((i,j^*+r))&=(Y'_1)_i-j^*-r+A_{Y_1}((i,j^*+r))+1=\\
&=(Y'_1)_i-(Y'_2)_i+(Y'_2)_i-j^*-r+A_{Y_1}((i,j^*))-[A_{Y_1}((i,j^*))-A_{Y_1}((i,j^*+r))]+1=\\
&=(Y'_1)_i-(Y'_2)_i+(Y'_2)_i-[A_{Y_1}((i,j^*))+(Y'_2)_i+1]-r+A_{Y_1}((i,j^*))-[A_{Y_1}((i,j^*))-A_{Y_1}((i,j^*+r))]+1=\\
&=(Y'_1)_i-(Y'_2)_i-r-[A_{Y_1}((i,j^*))-A_{Y_1}((i,j^*+r))],
\end{aligned}
\end{equation}
for every $0<r\le(Y'_1)_i-j^*$. Moreover,
\begin{equation}
[(Y'_1)_i-(Y'_2)_i]-h_{Y_1}((i,j^*+r))=r+[A_{Y_1}((i,j^*))-A_{Y_1}((i,j^*+r))].
\end{equation}
Since the quantity $[A_{Y_1}((i,j^*))-A_{Y_1}((i,j^*+r))]$ is positive, we have that the product over the boxes on the right of $(i,j^*)$ is bounded by
\begin{equation}\label{rightboxes}
\prod_{j=j^*+r}\frac{(Y'_1)_i-(Y'_2)_i-r}{r}=\prod_{r=1}^{(Y'_1)_i-j^*}\frac{(Y'_1)_i-(Y'_2)_i-r}{r}=\binom{(Y'_1)_i-(Y'_2)_i-1}{(Y'_1)_i-j^*}.
\end{equation}
Putting together \eqref{leftboxes} and \eqref{rightboxes}, the product over the boxes of the $i$th row of $Y_1$ is bounded by
\begin{equation}\label{productquantity}
\binom{(Y'_1)_i-(Y'_2)_i+j^*-1}{j^*-1}\binom{(Y'_1)_i-(Y'_2)_i-1}{(Y'_1)_i-j^*}.
\end{equation}

Since $j^*\le (Y'_1)_i$ and $j^*>(Y'_2)_i$, we have that
\begin{equation}
\begin{aligned}
&\binom{(Y'_1)_i-(Y'_2)_i+j^*-1}{j^*-1}\le \binom{2(Y'_1)_i-(Y'_2)_i-1}{j^*-1}\le 2^{2(Y'_1)_i-(Y'_2)_i},\\
&\binom{(Y'_1)_i-(Y'_2)_i-1}{(Y'_1)_i-j^*}\le 2^{(Y'_1)_i-(Y'_2)_i}.
\end{aligned}
\end{equation}

We conclude that
\begin{equation}
\binom{(Y'_1)_i-(Y'_2)_i+j^*-1}{j^*-1}\binom{(Y'_1)_i-(Y'_2)_i-1}{(Y'_1)_i-j^*}\le 2^{3(Y'_1)_i}.
\end{equation}
Considering the product of this bound for all the rows of $Y_1$, we can conclude
\begin{equation}
\prod_{(i,j)\in B_2(Y_1)}\bigg|\frac{h_{Y_1}((i,j))}{-h_{Y_1}((i,j))+(Y'_1)_i-(Y'_2)_i}\bigg|\le 8^{|Y_1|}.
\end{equation}

\section{Useful asymptotics}
\label{appe}

Here we collect some useful asymptotic formulae used in Section 4.

\begin{equation}
\begin{aligned}
\log(G(1-\mu+2(\alpha+n)))=&\frac{2n+2\alpha-\mu}{2}\log(2\pi)-\log(A)+\frac{1}{12}-\frac{3(2n)^2}{4}-(2\alpha-\mu)(2n)+\\
&+\left(\frac{(2n)^2}{2}-\frac{1}{12}+\frac{(2\alpha-\mu)^2}{2}+(2\alpha-\mu)(2n)\right)\log(2n)+\mathcal{O}\left(\frac{1}{n}\right);\\
\log(G(1-\mu-2(\alpha+n)))=&\log(G(1+\mu+2(\alpha+n)))-(
\mu+2(\alpha+n))\log(2\pi)+\int_0^{\mu+2(\alpha+n)}\pi z'\cot(\pi z')\mathrm{d}z'=\\
=&\frac{2n+2\alpha+\mu}{2}\log(2\pi)-\log(A)+\frac{1}{12}-\frac{3(2n)^2}{4}-(2\alpha+\mu)(2n)+\\
&+\left(\frac{(2n)^2}{2}-\frac{1}{12}+\frac{(2\alpha+\mu)^2}{2}+(2\alpha+\mu)(2n)\right)\log(2n)-(\mu+2(a+n))\log(2\pi)+\\
&+(\mu+2(\alpha+n))\log(1-\exp(2\pi i (\mu+2\alpha)))+\\
&-\frac{i\left(\pi^2(\mu+2(\alpha+n))^2+\mathrm{Li}_2(\exp(2\pi i (\mu+2\alpha)))\right)}{2\pi}+\mathcal{O}\left(\frac{1}{n}\right);\\
\log(G(1+2(\alpha+n)))=&\frac{2n+2\alpha}{2}\log(2\pi)-\log(A)+\frac{1}{12}-\frac{3(2n)^2}{4}-(2\alpha)(2n)+\\
&+\left(\frac{(2n)^2}{2}-\frac{1}{12}+\frac{(2\alpha)^2}{2}+(2\alpha)(2n)\right)\log(2n)+\mathcal{O}\left(\frac{1}{n}\right);\\
\log(G(1-2(\alpha+n)))=&\log(G(1+2(\alpha+n)))-(2(\alpha+n))\log(2\pi)+\int_0^{2(\alpha+n)}\pi z'\cot(\pi z')\mathrm{d}z'=\\
=&\frac{2n+2\alpha}{2}\log(2\pi)-\log(A)+\frac{1}{12}-\frac{3(2n)^2}{4}-(2\alpha)(2n)+\\
&+\left(\frac{(2n)^2}{2}-\frac{1}{12}+\frac{(2\alpha)^2}{2}+(2\alpha)(2n)\right)\log(2n)-(2(\alpha+n))\log(2\pi)+\\
&+(2(\alpha+n))\log(1-\exp(4\pi i \alpha))+\\
&-\frac{i\left(\pi^2(2(\alpha+n))^2+\mathrm{Li}_2(\exp(4\pi i \alpha))\right)}{2\pi}+\mathcal{O}\left(\frac{1}{n}\right).
\end{aligned}
\end{equation}
Hence,
\begin{equation}
\begin{aligned}
&\log(G(1-\mu-2(\alpha+n)))+\log(G(1-\mu+2(\alpha+n)))-\log(G(1+2(\alpha+n)))-\log(G(1-2(\alpha+n)))=\\
&=2n\log(\frac{1-\exp(2\pi i(\mu+2\alpha))}{1-\exp(4\pi i \alpha)})-2\pi i\mu n+\mu^2\log(2n)-\mu\log(2\pi)+2\alpha\log(\frac{1-\exp(2\pi i(\mu+2\alpha))}{1-\exp(4\pi i \alpha)})+\\
&+\mu\log(1-\exp(2\pi i(\mu+2\alpha)))-\frac{\mathrm{Li}_2(\exp(2\pi i(\mu+2\alpha)))-\mathrm{Li}_2(\exp(4\pi i\alpha))}{2\pi}-\frac{i\pi}{2}(\mu^2+4\alpha\mu)+\mathcal{O}\left(\frac{1}{n}\right).
\end{aligned}
\end{equation}
Therefore, up to $1/n$ corrections,
\begin{equation}
\frac{G(1-\mu-2(\alpha+n))G(1-\mu+2(\alpha+n))}{G(1+2(\alpha+n))G(1-2(\alpha+n))}
\propto
(2n)^{\mu^2}
\left(\frac{\sin(\pi (\mu+2\alpha))}{\sin(2\pi \alpha)}\right)^{2n}
\end{equation}
where the proportionality constant is independent on $n$.

\newpage
 	
 	\color{black}

		\bibliographystyle{unsrtaipauth4-1}
		\bibliography{Biblio}

	\end{document}